\documentclass[12pt, onecolumn,a4paper,accepted=2026-10-07]{quantumarticle}
\pdfoutput=1 
\usepackage[T1]{fontenc}
\usepackage{tikz}
\usetikzlibrary{positioning, calc}
\usetikzlibrary{arrows}
\usepackage{tikz-3dplot}
\usepackage{graphicx}
\usepackage{mdwlist}
\usepackage{color}
\usepackage{amssymb}
\usepackage{physics}
\usepackage{amsmath}
\usepackage{array}
\usepackage{doi}
\usepackage{url,hyperref}
\usepackage{mathtools}
\usepackage{dsfont}
\usepackage{makecell}
\usepackage{qcircuit}
\usepackage{mleftright}
\mleftright
\usepackage{bbm}
\usetikzlibrary{decorations.pathmorphing}
\usepackage{enumitem}
\usepackage[numbers,sort&compress]{natbib}
\usetikzlibrary{fadings}
\usetikzlibrary{decorations.pathmorphing}
\usepackage{mathrsfs}
\usepackage{authblk}
\usepackage[mathscr]{euscript}
\usepackage{enumitem}
\usepackage{cleveref}
\usepackage{slashed}
\usepackage{float} 
\usepackage{subcaption}

\DeclareMathAlphabet\mathbfcal{OMS}{cmsy}{b}{n}
  
  \makeatletter
  \def\namedlabel#1#2{\begingroup
	#2%
	\def\@currentlabel{#2}%
     	\phantomsection\label{#1}\endgroup
  }
  \makeatother

  \definecolor{QuSoft}{RGB}{179, 16, 32}

  \newcommand{\dens}{\mathcal{D}} 
  \newcommand{\nn}{\nonumber}
  
\usepackage{braket}
  
  \makeatletter
\renewcommand*{\@fnsymbol}[1]{\ensuremath{\ifcase#1\or *\or \mathparagraph\or \ddagger\or
    \mathsection \or \dagger\or \|\or **\or \dagger\dagger
    \or \ddagger\ddagger \else\@ctrerr\fi}}
\makeatother

  \newcommand{\ie}{{i.e.}}

\usetikzlibrary{decorations.pathreplacing}
\tikzset{snake it/.style={decorate, decoration=snake}}

\usetikzlibrary{decorations.pathreplacing,decorations.markings}

\tikzset{
    >=stealth',
    punkt/.style={
           rectangle,
           rounded corners,
           draw=black, very thick,
           text width=6.5em,
           minimum height=2em,
           text centered},
    pil/.style={
           ->,
           thick,
           shorten <=2pt,
           shorten >=2pt,},
  on each segment/.style={
    decorate,
    decoration={
      show path construction,
      moveto code={},
      lineto code={
        \path [#1]
        (\tikzinputsegmentfirst) -- (\tikzinputsegmentlast);
      },
      curveto code={
        \path [#1] (\tikzinputsegmentfirst)
        .. controls
        (\tikzinputsegmentsupporta) and (\tikzinputsegmentsupportb)
        ..
        (\tikzinputsegmentlast);
      },
      closepath code={
        \path [#1]
        (\tikzinputsegmentfirst) -- (\tikzinputsegmentlast);
      },
    },
  },
  mid arrow/.style={postaction={decorate,decoration={
        markings,
        mark=at position .5 with {\arrow[#1]{stealth'}}
      }}}
}

\newtheorem{theorem}{Theorem}

\newtheorem{definition}[theorem]{Definition}

\newtheorem{lemma}[theorem]{Lemma}

\newtheorem{proposition}[theorem]{Proposition}

\newenvironment{proof}[1][Proof]{\textbf{#1.}}{\ \rule{0.5em}{0.5em}}
\begin{document}
\numberwithin{equation}{section}

\title{On the Role of Quantum Communication and Loss in Attacks on Quantum Position Verification}

\author[1,2]{Rene Allerstorfer}
\orcid{0000-0002-1837-7725}

\author[2,3]{Harry Buhrman}
\orcid{0009-0000-5306-6853}

\author[2,3]{Florian Speelman}
\orcid{0000-0003-3792-9908}

\author[1,2,4]{Philip Verduyn Lunel}
\email{philip.verduyn-lunel@lip6.fr}
\orcid{0000-0001-5419-6027}

\affiliation[1]{CWI, Amsterdam, The Netherlands}
\affiliation[2]{QuSoft, Amsterdam, The Netherlands}
\affiliation[3]{University of Amsterdam, Amsterdam, The Netherlands}
\affiliation[4]{Sorbonne Universit\'e, CNRS, LIP6, Paris, France}

\maketitle

\begin{abstract}
\noindent We study the role of quantum communication in attacks on quantum position verification. In this work, we construct the first known example of a QPV protocol that is provably secure against unentangled attackers restricted to classical communication, but can be perfectly attacked by local operations and a single round of simultaneous quantum communication, indicating that allowing for quantum communication during the execution of the protocol may break security in a different way than just by establishing shared entanglement. We also show that any QPV protocol with a positive security gap against classical communication can be transformed into a protocol with an explicit positive gap between the honest prover’s success probability and that of initially unentangled attackers using a single round of simultaneous quantum communication. We further show, using arguments based on the monogamy of entanglement, that the task of Bell state discrimination cannot be done locally with a single round of quantum communication, not even probabilistically (when we allow attackers to say loss sometimes), making this the first fully loss-tolerant QPV task secure against quantum communication attacks. Finally, we observe that any multi-round QPV protocol can be attacked with a linear amount of entanglement if the loss is high enough.
\end{abstract}


\section{Introduction}
Geographical position is an important contributor to trust---for example, a message which provably comes from a secure location in a government institution, has automatic credence to actually be sent by that government. Position-based cryptography is the study of using position as a cryptographic credential. The most basic task here is to certify someone's position, but this can be extended to messages that can only be read at a certain location, or to \emph{authenticating} that a message came (unaltered) from a certain location.

We will focus on the task of \emph{position verification}, which can be used as the building block for tasks like position-based authentication. For simplicity, the focus will be on the one-dimensional case, i.e.\ verifying one's position on a line, but the relevant ideas generalize readily to more dimensions. In our case, protocols will have the form of two \emph{verifiers}, $\mathsf{V_A}$ and $\mathsf{V_B}$, attempting to verify the location of a \emph{prover} $\mathsf{P}$.
An adversary to a scheme will take the form of a coalition of attackers, while the location of $\mathsf{P}$ is empty.
We'll use $\mathsf{A}$ (or Alice) for the attacker located between $\mathsf{V_A}$ and the location of $\mathsf{P}$, and $\mathsf{B}$ (or Bob) for the attacker location between the location of $\mathsf{P}$ and $\mathsf{V_B}$. Throughout the paper, we work in this two-verifier model and take all protocol answers to be classical. All quantum registers are finite-dimensional.

It was shown by Chandran, Goyal, Moriarty, and Ostrovsky~\cite{chandran_position_2009} that without any additional assumptions, position verification is an impossible task to achieve classically. The quantum study of quantum position verification (QPV) was first initiated by Beausoleil, Kent, Munro, and Spiller resulting in a patent published in 2006 \cite{KentPatent2006}. The topic first appeared in the academic literature in 2010 ~\cite{Malaney2010, MalaneyNoisy2010}, followed by various proposals and ad-hoc attacks ~\cite{kent_quantum_2011,lau_insecurity_2011}. A general attack on quantum protocols for this task was presented by Buhrman, Chandran, Fehr, Gelles, Goyal, Ostrovsky, and Schaffner~\cite{buhrman_position-based_2011}, requiring a doubly-exponential amount of entanglement. This attack was further improved to requiring an exponential amount of entanglement by Beigi and K\"onig~\cite{beigi_simplified_2011} -- much more efficient but still impractically large. (See also~\cite{gao2016quantum,dolev2019} for generalizations of such attacks to different settings, with similar entanglement scaling.)

Additionally, other protocols have been proposed \cite{kent_quantum_2011,chakraborty_practical_2015, unruh_quantum_2014,junge2021geometry,bluhm2021position}, that combine classical and quantum information in interesting ways, sometimes requiring intricate methods to attack~\cite{buhrman2013garden, speelman2016,olivo_breaking_2020}. 

In our contribution we present several results on the role of quantum communication in QPV attacks.

\paragraph{The role of quantum communication for attacks on QPV.}
Some works \cite{buhrman_near-optimal_2011,tomamichel_monogamy--entanglement_2013,buhrman2013garden,bluhm2021position}  attempt to lower bound the pre-shared entanglement required from attackers that are allowed a round of simultaneous \emph{quantum} communication, while other results, such as~\cite{beigi_simplified_2011,ribeiro_tight_2015,qi_loss-tolerant_2015,qi_free-space_2015,lim_loss-tolerant_2016,gonzales_bounds_2020,olivo_breaking_2020} assume attackers that are restricted to communicate only classically.

Even though quantum communication can potentially be simulated by teleportation, it is not immediately clear how to compare bounds between these two settings, especially in case where the exact size of the lower bound is of interest.\footnote{For instance, the bound of \cite{ribeiro_tight_2015} does not fully supersede  \cite{tomamichel_monogamy--entanglement_2013}, and therefore finding a tight lower bound for the parallel-repetition of the QPV$_{\textsf{BB84}}$ protocol against attackers that have quantum communication remains open.}
The simplest version of this question can be asked for unentangled attackers: If a (quantum-question, classical-reply) QPV protocol is secure against unentangled attackers that communicate classically, is that protocol also secure against unentangled attackers that are allowed to use quantum communication?

To that end, we present the following results:

\begin{itemize}
    \item First, we answer the above question in the negative: We construct a protocol that is provably secure against unentangled attackers that can use classical communication, but can be broken by a single round of simultaneous quantum communication. This shows that some care has to be taken when interpreting results that restrict to classical messages only.
    
    \item Interestingly, we are additionally able to show that our counter-example is in some sense artificial: Given a protocol that is secure against classical messages, but insecure when quantum communication is allowed, it is always possible to transform this protocol into one that is secure when quantum communication is allowed. A positive honest-versus-LOCC gap implies that one of the protocols obtained by finitely many local splitting maps has an explicit positive honest-versus-LOQC gap.
    
    This new protocol can be constructed from the given protocol by applying local maps to the messages from the verifiers $\mathsf{V_A}$, $\mathsf{V_B}$, without having to modify the output predicate. Our proof for this statement involves a recursive argument, where we view the states after quantum communication of a successful attack as the input messages to two new protocols. We then recursively consider an increasing number of new possible protocols, and use \emph{emergent classicality}~\cite{qi2020emergent} to show that a secure protocol of the required form has to exist.
    
    \item We proceed by considering the task of Bell state discrimination\footnote{Where the input is a randomly chosen Bell state.} and prove that this task cannot be done perfectly with only local operations and one round of simultaneous communication. The proof relies on new arguments based on the monogamy of entanglement. We consider a purified version of this task in the QPV setting (delaying the honest measurement to the end of the protocol) and show that the squashed entanglement \cite{christandl2004squashed} of the state $\rho_{V_A V_B}$, on which the honest Bell measurement is applied to, is upper bounded by $E_{sq}(V_A : V_B)_\rho \leq 1/2$. Hence attackers won't be able to perfectly predict the honest result. To get an explicit upper bound on the attack success probability with quantum communication $p_\text{succ}^\text{qc}$ we use the \textit{hashing bound} from \cite{devetak2005distillation} which lower bounds the squashed entanglement, allowing us to upper bound a parameter that leads to $p_\text{succ}^\text{qc} \leq 0.926$. We further improve this bound to $p_\text{succ}^\text{qc} \leq 3/4$ via a different argument based on the no-cloning theorem.
    
    We additionally show that even in the lossy scenario it remains that $p_\text{succ}^\text{qc}(\eta) < 1$ for any transmission rate $0 < \eta \leq 1$. Under the common-loss convention formalized in Definition~\ref{def:lossy-qpv-bell}, this makes the task of Bell state discrimination, and by implication the QPV protocol based on the SWAP-test~\cite{allerstorfer2022practical}, the first fully loss-tolerant
    QPV protocol that remains secure in the setting where attackers are
    allowed quantum communication.
\end{itemize}

 Finally, we present a result relating loss tolerance and entanglement attacks in QPV:

\begin{itemize}
    \item We observe that in a setting with loss any multi-round QPV protocol can be broken with only a linear amount of pre-shared entanglement if the loss rate is high enough. In that sense, creating a fully loss tolerant QPV protocol which requires superlinear entanglement (in the number of qubits involved) is impossible. This follows directly from a simple observation: if there is no limit to the loss, the adversaries can attempt quantum teleportation and guess the teleportation corrections, claiming `loss' if the guess is incorrect.
\end{itemize}

The aspects of loss and quantum communication are practically very relevant, since in realistic settings loss rates will be high and, although attackers are restricted to only one round of simultaneous communication due to the timing constraints of QPV, they could in principle be able to quantum communicate and this might give them an advantage.

\subsection{Structure of the paper}

In section \ref{sec:quantumcomm} we present the first QPV protocol that is provably secure against attackers restricted to classical communication but broken by a single round of quantum communication. In section \ref{sec:splitting}, we prove quantitatively that a positive honest-versus-LOCC gap can be converted into a positive honest-versus-LOQC gap by finitely many local splitting maps. In section \ref{sec:QCproofagainstBellstatediscrimination} we show that the task of Bell state discrimination is secure against attackers who are allowed to use quantum communication. Extending this result, section \ref{sec:LossyQuantumCommunicationattackbound} rules out a perfect positive-rate attack under the common-loss convention of Definition~\ref{def:lossy-qpv-bell}. Finally, in section \ref{sec:losstolerance} we make the observation that allowing loss for pre-shared entangled attackers allows for any protocol with high enough loss rate to be broken with linear entanglement.

\section{Preliminaries} \label{sec:prelim}
\subsection{Notation}
We denote parties in QPV protocols by letters \textsf{A}, \textsf{B}, etc. and their quantum registers as $A_1 \cdots A_n$, $B_1 \cdots B_n$ and so on, respectively. Sometimes we may refer to ``all registers party $\mathsf{X}$ holds'' just by \textsf{X}, giving expression like $\operatorname{Pos}(\mathsf{A} \otimes \mathsf{B})$, for example. Unless otherwise indicated, $\lVert \cdot \rVert_p$ is the usual $p$-norm. The diamond norm on quantum channels is denoted by $\lVert \cdot \rVert_\diamond$ and is defined as $\lVert \mathcal{C} \rVert_\diamond \coloneqq \max_{\rho} \lVert (\mathcal{C} \otimes \mathbbm{1}_k)(\rho) \rVert_1$ for quantum states $\rho$. Partial transposition of an operator $P$ with respect to party $\mathsf{B}$ is denoted $P^{T_\mathsf{B}}$. The set of positive partial transpose (PPT) measurements, that is, sets of positive semi-definite operators adding up to the identity, whose partial transposes are positive semi-definite as well, on two subsystems held by parties \textsf{A} and \textsf{B}, respectively, is PPT$(\mathsf{A}:\mathsf{B})$. We use the term ``Local Operations and Broadcast Communication'' (LOBC) to describe the scenario of a single round of simultaneous classical communication with local quantum operations before and after the round of communication. Finally, the image of a function $f$ is denoted by $\operatorname{Im}(f)$ and for a set $X$ we write $| X |$ for its cardinality. All other notation is explained in the text.

\subsection{Quantum position verification}
For simplicity, we treat the one-dimensional case here, where all parties are located on a line. The time needed to implement local operations is considered negligibly short compared to the time span of the entire protocol. In order to verify the position of an untrusted party $\mathsf{P}$, two trusted and spatially separated verifiers $\mathsf{V_A}, \mathsf{V_B}$ send quantum inputs to $\mathsf{P}$ from each side and ask them to apply a specific quantum operation. $\mathsf{P}$ has to apply the operation and respond immediately. In the end the verifiers check if they received an answer in time and consistent with the input and the demanded task. The attack model is as follows. Attackers trying to break the protocol are \textit{not} located at \textsf{P} but want to convince the verifiers that they are. Two attackers\footnote{The scenario of more attackers can be reduced to the one described above. Indeed, the attackers closest to \textsf{P} could simply simulate all the other attackers themselves.} $\mathsf{A, B}$ can position themselves between $\mathsf{V_A}, \mathsf{P}$ and $\mathsf{V_B}, \mathsf{P}$, respectively, and intercept the inputs, act locally, communicate one message to each other and then act locally again before they have to commit to answers $\tau_A, \tau_B$. They hence have to simulate the honest quantum operation using only local actions and 1 round of simultaneous communication. In general, they could also pre-share an entangled resource state $\eta_{AB}$ at the start of the protocol. This situation is depicted in figure~\ref{fig:general_qpv}.

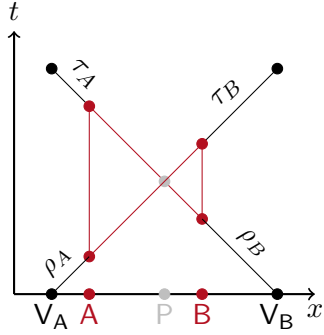
\begin{figure}[h]
\centering
\begin{tikzpicture}
	\draw[->,black,thick] (-0.5,-3) -- (-0.5,0.5) node[above]{$t$};
	\draw[->,black,thick] (-0.5,-3) -- (3.5,-3) node[below]{$x$};
	
    \filldraw [black] (0,0) circle (2pt);
    \draw[black]  (0,0) -- (0.5,-0.5) node[midway, above, sloped]{$\tau_A$};
    \filldraw [QuSoft] (0.5,-0.5) circle (2pt);
    \draw[QuSoft]  (0.5,-0.5) -- (2,-2);
    \draw[QuSoft]  (0.5,-0.5) -- (0.5,-2.5);
    \filldraw [QuSoft] (0.5,-2.5) circle (2pt);
    \draw[black]  (0.5,-2.5) -- (0,-3) node[midway, above, sloped]{\small $\rho_A$};
    \filldraw [black] (0,-3) circle (2pt) node[below]{$\mathsf{V_A}$};
    
    \filldraw [lightgray] (1.5,-1.5) circle (2pt);
    
    \filldraw [black] (3,0) circle (2pt);
    \draw[black]  (3,0) -- (2,-1) node[midway, above, sloped]{$\tau_B$};
    \filldraw [QuSoft] (2,-1) circle (2pt);
    \draw[QuSoft]  (2,-1) -- (0.5,-2.5);
    \draw[QuSoft]  (2,-1) -- (2,-2);
    \filldraw [QuSoft] (2,-2) circle (2pt);
    \draw[black]  (2,-2) -- (3,-3) node[midway, above, sloped]{\small $\rho_B$};
    \filldraw [black] (3,-3) circle (2pt) node[below]{$\mathsf{V_B}$};
    
    \filldraw [QuSoft] (0.5,-3) circle (2pt) node[below]{$\mathsf{A}$};
    \filldraw [lightgray] (1.5,-3) circle (2pt) node[below]{$\mathsf{P}$};
    \filldraw [QuSoft] (2,-3) circle (2pt) node[below]{$\mathsf{B}$};
    
\end{tikzpicture}
\caption{Space-time diagram of a general QPV protocol. We assume all information travels at the speed of light. For graphical simplicity we have put $\mathsf{P}$ exactly in the middle of $\mathsf{V_A}$ and $\mathsf{V_B}$ (which is not necessary). The attackers, not being at position $\mathsf{P}$, would like to convince the verifiers that they are at $\mathsf{P}$ by simulating the honest operation via local operations and one round of simultaneous communication.}
\label{fig:general_qpv}
\end{figure}

\section{QPV and quantum communication}\label{sec:moreresults}
\subsection{A protocol for which quantum communication gives an advantage over LOCC}\label{sec:quantumcomm}
A natural question one might ask is whether there is any advantage for attackers in QPV protocols if they are allowed to perform local operations and quantum communication (LOQC) instead of classical communication. In what follows we will construct an explicit example of a QPV protocol with classical outputs where there is a finite gap in success probability for LOQC strategies over LOCC strategies. 

First, consider the protocol where two verifiers both send half of either one randomly picked symmetric Bell state $\{\ket{\Phi^+}, \ket{\Phi^-}, \ket{\Psi^+}\}$ or the antisymmetric Bell state $\ket{\Psi^-}$, and ask an honest prover whether the entangled state they have sent is symmetric or antisymmetric. An honest prover who can apply entangling operations can answer this question with success probability $1$ by applying a SWAP-test \cite{buhrman_quantum_2001} on the state. From the analysis of the corresponding semidefinite program (SDP) optimized over PPT measurements it turns out that the best LOCC strategy is upper bounded by $5/6$ (see Appendix \ref{AppendixQCbetter}). The LOCC strategy of measuring both qubits in the computational basis and answering the XOR of the outcomes attains this success probability, so the upper bound over PPT measurements is attained by a LOCC measurement. 

Now suppose the verifiers execute two instances of the preceding protocol concurrently. In the symmetric case, each instance contains an independently sampled symmetric Bell state; in the antisymmetric case, both instances contain $\ket{\Psi^-}$. The prover must decide which of these two cases occurred. An honest prover who can apply entangling operations can still solve this protocol with success probability $1$ by applying a SWAP-test to one of the two pairs. Now note that attackers who have access to a quantum channel can send half of their input state to each other such that both attackers locally end up with a Bell state which they can perfectly determine. Thus attackers restricted to quantum communication can attack this protocol perfectly. Interestingly, it turns out that this is not possible for attackers restricted to classical communication. 

From the analysis of the SDP it turns out that the upper bound for two attackers restricted to PPT measurements is $17/18$, cf. appendix \ref{AppendixQCbetter}. Again there is a LOCC strategy that makes this bound tight, namely measuring both pairs in the computational basis and only answering ``antisymmetric'' if both pairs have unequal measurement outcomes and respond ``symmetric'' otherwise. This strategy is always correct on antisymmetric inputs. And it is only incorrect on symmetric inputs if both states are $\ket{\Psi^+}$, which happens with probability $1/18$, so the total probability of success of the LOCC protocol then becomes $17/18$. By incorporating loss in the SDP program as done in \cite{lim_loss-tolerant_2016} and \cite{allerstorfer2022practical}, we also find that this protocol is loss-tolerant.

\begin{proposition}[Quantum communication strictly outperforms classical communication]
\label{prop:strict-qc-advantage}
For the concurrent two-instance symmetric-versus-antisymmetric protocol, the honest value and the optimal value of initially unentangled attackers using one round of simultaneous quantum communication are both equal to $1$, whereas the optimal one-round LOCC value is $17/18$.
\end{proposition}

Thus we have constructed a QPV protocol where the probability of success for attackers restricted to single round LOCC measurements is strictly lower than attackers restricted to single round LOQC measurements. This shows that there can be an advantage for quantum communication over classical communication, and it could be important in the analysis of the security of QPV protocols. However, it is clear that our construction is not a very good protocol as there is redundant information given to the attackers and sending just one of the two symmetric or antisymmetric states would give a seemingly better protocol.

\subsection{Splitting Scheme}
\label{sec:splitting}

In this section we present a construction that shows the existence of a QPV protocol secure against attackers using a single round of simultaneous quantum communication, starting from a QPV protocol that is secure against adversaries restricted to LOCC operations. We use the fact that a sufficiently successful quantum communication attack on a QPV protocol generates two new QPV protocols. Applying this construction recursively ultimately leads to a QPV protocol that is secure against adversaries restricted to LOQC.

Take any QPV protocol in which two verifiers $\mathsf{V_A},\mathsf{V_B}$ send quantum inputs $A,B$ and ask for the same classical answer to a test on the joint input. Suppose the protocol is secure against adversaries restricted to LOCC, \ie, there is a finite gap between the success probability of an honest prover and that of adversaries restricted to LOCC operations. We show that this implies the existence of a derived QPV protocol whose inputs are obtained by applying local quantum channels to the original inputs and for which, for some explicit $\delta>0$, the success probability of attackers using quantum communication is at most the honest success probability minus $\delta$. The derived protocol also remains secure against LOCC attacks.

In other words, insecurity against quantum communication can arise from the inputs being, in a way, \textit{overinformative}. A successful quantum communication attack shows how the inputs can be split into locally processed shares, which define two new QPV protocols. By applying this construction recursively, we can achieve security against quantum communication, provided that the original protocol was secure against LOCC strategies.

We now prove a quantitative version of this splitting argument. Consider a QPV protocol $\mathcal P$ with finite-dimensional input registers $A$ and $B$. A reference register $R$, retained by the verifiers, contains any purification, classical question, and information needed for the final acceptance test. Thus the verifier preparation may be an arbitrary state
\begin{equation}
  \rho_{RAB}\in\dens(R\otimes A\otimes B),
\end{equation}
and in particular the input on $AB$ need not be separable. As in the protocols considered here, both verifiers accept only when they receive the same valid answer.

For a protocol $\mathcal P$, let $p_{\mathrm{hon}}(\mathcal P)$ be its honest success probability, let $\omega_{\rm LOCC}(\mathcal P)$ be the supremal success probability of attackers using local operations and one simultaneous round of classical communication, and let $\omega_{\rm LOQC}(\mathcal P)$ be the corresponding value for initially unentangled attackers using one simultaneous round of quantum communication. Let $\mathcal P_{\varnothing}$ be the initial protocol and set
\begin{equation}
  p_{\mathrm{hon},0}:=p_{\mathrm{hon}}(\mathcal P_{\varnothing}),\qquad
  c_0:=\omega_{\rm LOCC}(\mathcal P_{\varnothing}),\qquad
  \Delta:=p_{\mathrm{hon},0}-c_0>0.
  \label{eq:revised-initial-gap}
\end{equation}
In other words, $\Delta$ is the gap between the success probability of the honest prover and LOCC attackers. Write $d_A:=\dim(A)$ and $d_B:=\dim(B)$. An unentangled one-round quantum-communication attack first applies local
channels
\begin{equation}
  \mathcal A:A\longrightarrow A_0A_1,
  \qquad
  \mathcal B:B\longrightarrow B_0B_1.
\end{equation}
Attacker $\mathsf A$ retains $A_0$ and sends $A_1$, while attacker $\mathsf B$ sends $B_0$ and retains $B_1$. After the simultaneous exchange, their local views are $A_0B_0$ and $A_1B_1$, respectively. We reinterpret these two locally processed views as the inputs of two derived protocols, rather than as attacks on those protocols. For $b\in\{0,1\}$, the verifier preparation of $\mathcal P_b$ is the reduced state on $RA_bB_b$. The honest prover applies attacker $b$'s final POVM to $A_bB_b$ and sends the resulting label $z$ to both verifiers. The verifiers retain $R$ and apply the original public acceptance predicate to the common label $z$. Hence $p_{\mathrm{hon}}(\mathcal P_b)$ is exactly the probability that attacker $b$'s local answer passes the original verifier test when the other attacker's answer is ignored.

Thus, a successful quantum-communication attack on a parent protocol produces two new honest-prover tasks whose honest success is at least the parent attack's success; it does not produce an attack on either child. The following lemma records the two properties that make this construction iterable: honest success is lower-bounded by the parent attack's success, whereas the LOCC value cannot increase.

\begin{lemma}[Attack-induced splitting]
\label{lem:revised-one-splitting-step}
Suppose an initially unentangled one-round LOQC attack on $\mathcal P$
succeeds with probability $s$, and let $\mathcal P_0,\mathcal P_1$ be its
induced child protocols. Then, for $b\in\{0,1\}$,
\begin{equation}
  p_{\mathrm{hon}}(\mathcal P_b)\geq s,
  \qquad
  \omega_{\rm LOCC}(\mathcal P_b)
  \leq\omega_{\rm LOCC}(\mathcal P).
  \label{eq:revised-child-bounds}
\end{equation}
\end{lemma}

\begin{proof} \
Global success of the parent attack implies that attacker $b$'s local answer
is valid. The probability of the latter event is precisely the honest
success probability of $\mathcal P_b$, proving the first inequality.

Conversely, an LOCC attack on $\mathcal P_b$ gives an LOCC attack on $\mathcal P$: the parent attackers apply $\mathcal A$ and $\mathcal B$, discard the registers belonging to the other child, and execute the child attack. This uses only local preprocessing and the one simultaneous classical message of the child attack, proving the second inequality.
\end{proof}

For $n\geq2$, define
\begin{align}
  \varepsilon_{A,n}
  &\coloneq d_A^3\sqrt{\frac{2\ln(d_A)}{2^{n-1}-1}},
  &
  \varepsilon_{B,n}
  &\coloneq d_B^3\sqrt{\frac{2\ln(d_B)}{2^{n-1}-1}}.
  \label{eq:revised-epsilons}
\end{align}
We consider any $n$ for which
\begin{equation}
  \frac{\varepsilon_{A,n}+\varepsilon_{B,n}}{2}<\Delta,
  \label{eq:revised-depth-condition}
\end{equation}
and set
\begin{equation}
  \delta_n
  :=\frac{\displaystyle
      \Delta-\frac{\varepsilon_{A,n}+\varepsilon_{B,n}}{2}}
     {2n}>0.
  \label{eq:revised-delta}
\end{equation}
Typically, we will be interested in having $\delta_n$ as large as possible. Since the numerator and denominator increase strictly in $n$ there is no universal value of $n$ that is optimal. The optimal choice depends on the parameters of the protocol and can be found by a simple numerical optimization over the admissible integers $n$.

We now iterate the preceding construction using a stop-or-split rule. At a node $v$, we first ask whether
$$
  \omega_{\rm LOQC}(\mathcal P_v)
  \leq p_{\mathrm{hon}}(\mathcal P_v)-\delta_n.
$$
If so, every one-round LOQC attack is at least $\delta_n$ worse than the
honest prover, and we stop that branch. Otherwise, because
$\omega_{\rm LOQC}(\mathcal P_v)$ is a supremum, we choose a sufficiently
successful LOQC attack with
$$
  s_v>p_{\mathrm{hon}}(\mathcal P_v)-\delta_n
$$
and use its local splitting channels and final POVMs to define the two child
protocols $\mathcal P_{v0}$ and $\mathcal P_{v1}$. By
Lemma~\ref{lem:revised-one-splitting-step}, for $b\in\{0,1\}$,
$$
  p_{\mathrm{hon}}(\mathcal P_{vb})
  >p_{\mathrm{hon}}(\mathcal P_v)-\delta_n,
  \qquad
  \omega_{\rm LOCC}(\mathcal P_{vb})
  \leq\omega_{\rm LOCC}(\mathcal P_v).
$$
We repeat the same test at every child that has not stopped.

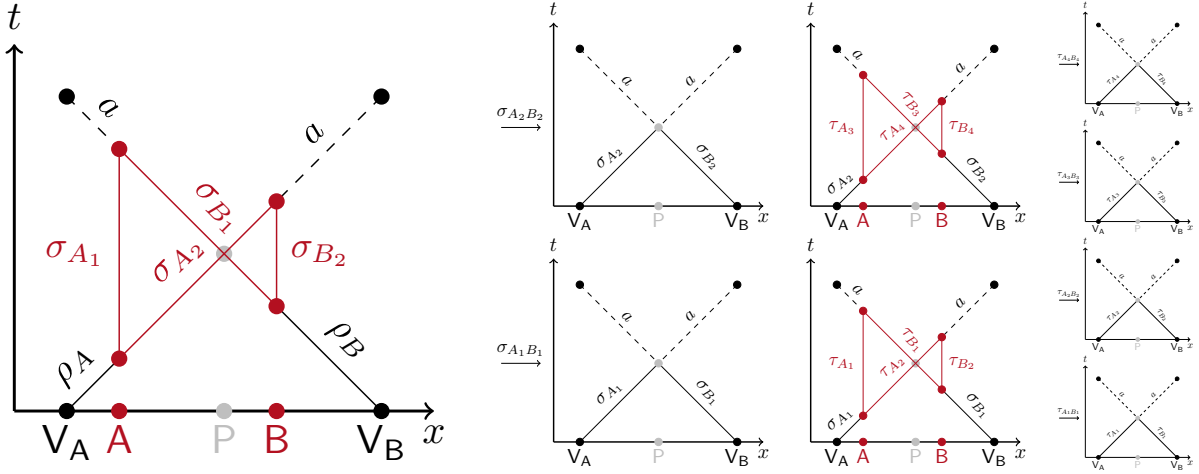
\begin{figure}[ht]
\captionsetup{width=.8\linewidth}
\centering
\resizebox{\linewidth}{!}{
\scalebox{1.2}{
\begin{tikzpicture}
	\draw[->,black,thick] (-0.5,-3) -- (-0.5,0.5) node[above]{$t$};
	\draw[->,black,thick] (-0.5,-3) -- (3.5,-3) node[below]{$x$};
    \filldraw [lightgray] (1.5,-3) circle (2pt) node[below]{$\mathsf{P}$};	
    \filldraw [lightgray] (1.5,-1.5) circle (2pt);
	
    \filldraw [black] (0,0) circle (2pt);
    \draw[black,dashed]  (0,0) -- (0.5,-0.5) node[midway, above, sloped]{\small $a$};
    \draw[black]  (0.5,-2.5) -- (0,-3) node[midway, above, sloped]{\small $\rho_A$};
    \filldraw [QuSoft] (0.5,-0.5) circle (2pt);
    \draw[QuSoft]  (0.5,-0.5) -- (2,-2) node[midway, above, sloped]{\small $\sigma_{B_1}$};
    \draw[QuSoft]  (0.5,-0.5) -- (0.5,-2.5) node[midway, left]{\small $\sigma_{A_1}$};
    \filldraw [QuSoft] (0.5,-2.5) circle (2pt) ;
    \filldraw [black] (0,-3) circle (2pt) node[below]{$\mathsf{V_A}$};
    
    \filldraw [black] (3,0) circle (2pt);
    \draw[black,dashed]  (3,0) -- (2,-1) node[midway, above, sloped]{\small $a$};
    \draw[black]  (2,-2) -- (3,-3) node[midway, above, sloped]{\small $\rho_B$};
    \filldraw [black] (3,-3) circle (2pt) node[below]{$\mathsf{V_B}$};
    \filldraw [QuSoft] (2,-1) circle (2pt); 
    \draw[QuSoft]  (2,-1) -- (0.5,-2.5) node[midway, above, sloped]{\small $\sigma_{A_2}$};
    \draw[QuSoft]  (2,-1) -- (2,-2) node[midway, right]{\small $\sigma_{B_2}$};
    \filldraw [QuSoft] (2,-2) circle (2pt);
    
    \filldraw [QuSoft] (0.5,-3) circle (2pt) node[below]{$\mathsf{A}$};

    \filldraw [QuSoft] (2,-3) circle (2pt) node[below]{$\mathsf{B}$};
    
\end{tikzpicture}
}
\scalebox{0.6}{
\begin{tikzpicture}
    \draw[->, black] (-1.5, -1.5+4.5) -- (-0.75, -1.5+4.5) node[midway, above]{\small $\sigma_{A_2B_2}$};
	\draw[->,black,thick] (-0.5,-3+4.5) -- (-0.5,0.5+4.5) node[above]{$t$};
	\draw[->,black,thick] (-0.5,-3+4.5) -- (3.5,-3+4.5) node[below]{$x$};
    \filldraw [lightgray] (1.5,-3+4.5) circle (2pt) node[below]{$\mathsf{P}$};

    \filldraw [black] (0,0+4.5) circle (2pt);
    \draw[black,dashed]  (0,0+4.5) -- (1.5,-1.5+4.5) node[midway, above, sloped]{\small $a$};
    \draw[black]  (1.5,-1.5+4.5) -- (0,-3+4.5) node[midway, above, sloped]{\small $\sigma_{A_2}$};
    \filldraw [black] (0,-3+4.5) circle (2pt) node[below]{$\mathsf{V_A}$};
    
    \filldraw [black] (3,0+4.5) circle (2pt);
    \draw[black,dashed]  (3,0+4.5) -- (1.5,-1.5+4.5) node[midway, above, sloped]{\small $a$};
    \draw[black]  (1.5,-1.5+4.5) -- (3,-3+4.5) node[midway, above, sloped]{\small $\sigma_{B_2}$};
    \filldraw [black] (3,-3+4.5) circle (2pt) node[below]{$\mathsf{V_B}$};
    \filldraw [lightgray] (1.5,-1.5+4.5) circle (2pt);

    \draw[->, black] (-1.5, -1.5) -- (-0.75, -1.5) node[midway, above]{\small $\sigma_{A_1B_1}$};
	\draw[->,black,thick] (-0.5,-3) -- (-0.5,0.5) node[above]{$t$};
	\draw[->,black,thick] (-0.5,-3) -- (3.5,-3) node[below]{$x$};
    \filldraw [lightgray] (1.5,-3) circle (2pt) node[below]{$\mathsf{P}$};	
    
    \filldraw [black] (0,0) circle (2pt);
    \draw[black,dashed]  (0,0) -- (1.5,-1.5) node[midway, above, sloped]{\small $a$};
    \draw[black]  (1.5,-1.5) -- (0,-3) node[midway, above, sloped]{\small $\sigma_{A_1}$};
    \filldraw [black] (0,-3) circle (2pt) node[below]{$\mathsf{V_A}$};
    
    \filldraw [black] (3,0) circle (2pt);
    \draw[black,dashed]  (3,0) -- (1.5,-1.5) node[midway, above, sloped]{\small $a$};
    \draw[black]  (1.5,-1.5) -- (3,-3) node[midway, above, sloped]{\small $\sigma_{B_1}$};
    \filldraw [black] (3,-3) circle (2pt) node[below]{$\mathsf{V_B}$};
    \filldraw [lightgray] (1.5,-1.5) circle (2pt);
\end{tikzpicture}
}
\scalebox{0.6}{
\begin{tikzpicture}
	\draw[->,black,thick] (-0.5,-3+4.5) -- (-0.5,0.5+4.5) node[above]{$t$};
	\draw[->,black,thick] (-0.5,-3+4.5) -- (3.5,-3+4.5) node[below]{$x$};
    \filldraw [lightgray] (1.5,-3+4.5) circle (2pt) node[below]{$\mathsf{P}$};	
    \filldraw [lightgray] (1.5,-1.5+4.5) circle (2pt);
	
    \filldraw [black] (0,0+4.5) circle (2pt);
    \draw[black,dashed]  (0,0+4.5) -- (0.5,-0.5+4.5) node[midway, above, sloped]{\small $a$};
    \draw[black]  (0.5,-2.5+4.5) -- (0,-3+4.5) node[midway, above, sloped]{\small $\sigma_{A_2}$};
    \filldraw [QuSoft] (0.5,-0.5+4.5) circle (2pt); 
    \draw[QuSoft]  (0.5,-0.5+4.5) -- (2,-2+4.5) node[midway, above, sloped]{\small $\tau_{B_3}$};
    \draw[QuSoft]  (0.5,-0.5+4.5) -- (0.5,-2.5+4.5) node[midway, left]{\small $\tau_{A_3}$};
    \filldraw [QuSoft] (0.5,-2.5+4.5) circle (2pt) ;
    \filldraw [black] (0,-3+4.5) circle (2pt) node[below]{$\mathsf{V_A}$};
    
    \filldraw [black] (3,0+4.5) circle (2pt);
    \draw[black,dashed]  (3,0+4.5) -- (2,-1+4.5) node[midway, above, sloped]{\small $a$};
    \draw[black]  (2,-2+4.5) -- (3,-3+4.5) node[midway, above, sloped]{\small $\sigma_{B_2}$};
    \filldraw [black] (3,-3+4.5) circle (2pt) node[below]{$\mathsf{V_B}$};
    \filldraw [QuSoft] (2,-1+4.5) circle (2pt); 
    \draw[QuSoft]  (2,-1+4.5) -- (0.5,-2.5+4.5) node[midway, above, sloped]{\small $\tau_{A_4}$};
    \draw[QuSoft]  (2,-1+4.5) -- (2,-2+4.5) node[midway, right]{\small $\tau_{B_4}$};
    \filldraw [QuSoft] (2,-2+4.5) circle (2pt);
    
    \filldraw [QuSoft] (0.5,-3+4.5) circle (2pt) node[below]{$\mathsf{A}$};
    \filldraw [QuSoft] (2,-3+4.5) circle (2pt) node[below]{$\mathsf{B}$};

	\draw[->,black,thick] (-0.5,-3) -- (-0.5,0.5) node[above]{$t$};
	\draw[->,black,thick] (-0.5,-3) -- (3.5,-3) node[below]{$x$};
    \filldraw [lightgray] (1.5,-3) circle (2pt) node[below]{$\mathsf{P}$};	
    \filldraw [lightgray] (1.5,-1.5) circle (2pt);
	
    \filldraw [black] (0,0) circle (2pt);
    \draw[black,dashed]  (0,0) -- (0.5,-0.5) node[midway, above, sloped]{\small $a$};
    \draw[black]  (0.5,-2.5) -- (0,-3) node[midway, above, sloped]{\small $\sigma_{A_1}$};
    \filldraw [QuSoft] (0.5,-0.5) circle (2pt);
    \draw[QuSoft]  (0.5,-0.5) -- (2,-2) node[midway, above, sloped]{\small $\tau_{B_1}$};
    \draw[QuSoft]  (0.5,-0.5) -- (0.5,-2.5) node[midway, left]{\small $\tau_{A_1}$};
    \filldraw [QuSoft] (0.5,-2.5) circle (2pt) ;
    \filldraw [black] (0,-3) circle (2pt) node[below]{$\mathsf{V_A}$};
    
    \filldraw [black] (3,0) circle (2pt);
    \draw[black,dashed]  (3,0) -- (2,-1) node[midway, above, sloped]{\small $a$};
    \draw[black]  (2,-2) -- (3,-3) node[midway, above, sloped]{\small $\sigma_{B_1}$};
    \filldraw [black] (3,-3) circle (2pt) node[below]{$\mathsf{V_B}$};
    \filldraw [QuSoft] (2,-1) circle (2pt); 
    \draw[QuSoft]  (2,-1) -- (0.5,-2.5) node[midway, above, sloped]{\small $\tau_{A_2}$};
    \draw[QuSoft]  (2,-1) -- (2,-2) node[midway, right]{\small $\tau_{B_2}$};
    \filldraw [QuSoft] (2,-2) circle (2pt);
    
    \filldraw [QuSoft] (0.5,-3) circle (2pt) node[below]{$\mathsf{A}$};
    \filldraw [QuSoft] (2,-3) circle (2pt) node[below]{$\mathsf{B}$};
    
\end{tikzpicture}
}
\scalebox{0.3}{
\begin{tikzpicture}
    \draw[->, black] (-1.5, -1.5+13.5) -- (-0.75, -1.5+13.5) node[midway, above]{\small $\tau_{A_4B_4}$};
    \draw[->,black,thick] (-0.5,-3+13.5) -- (-0.5,0.5+13.5) node[above]{$t$};
	\draw[->,black,thick] (-0.5,-3+13.5) -- (3.5,-3+13.5) node[below]{$x$};
    \filldraw [lightgray] (1.5,-3+13.5) circle (2pt) node[below]{$\mathsf{P}$};	
    
    \filldraw [black] (0,0+13.5) circle (2pt);
    \draw[black,dashed]  (0,0+13.5) -- (1.5,-1.5+13.5) node[midway, above, sloped]{\small $a$};
    \draw[black]  (1.5,-1.5+13.5) -- (0,-3+13.5) node[midway, above, sloped]{\small $\tau_{A_4}$};
    \filldraw [black] (0,-3+13.5) circle (2pt) node[below]{$\mathsf{V_A}$};
    
    \filldraw [black] (3,0+13.5) circle (2pt);
    \draw[black,dashed]  (3,0+13.5) -- (1.5,-1.5+13.5) node[midway, above, sloped]{\small $a$};
    \draw[black]  (1.5,-1.5+13.5) -- (3,-3+13.5) node[midway, above, sloped]{\small $\tau_{B_4}$};
    \filldraw [black] (3,-3+13.5) circle (2pt) node[below]{$\mathsf{V_B}$};
    \filldraw [lightgray] (1.5,-1.5+13.5) circle (2pt);
    \draw[->, black] (-1.5, -1.5+9) -- (-0.75, -1.5+9) node[midway, above]{\small $\tau_{A_3B_3}$};

	\draw[->,black,thick] (-0.5,-3+9) -- (-0.5,0.5+9) node[above]{$t$};
	\draw[->,black,thick] (-0.5,-3+9) -- (3.5,-3+9) node[below]{$x$};
    \filldraw [lightgray] (1.5,-3+9) circle (2pt) node[below]{$\mathsf{P}$};	
    
    \filldraw [black] (0,0+9) circle (2pt);
    \draw[black,dashed]  (0,0+9) -- (1.5,-1.5+9) node[midway, above, sloped]{\small $a$};
    \draw[black]  (1.5,-1.5+9) -- (0,-3+9) node[midway, above, sloped]{\small $\tau_{A_3}$};
    \filldraw [black] (0,-3+9) circle (2pt) node[below]{$\mathsf{V_A}$};
    
    \filldraw [black] (3,0+9) circle (2pt);
    \draw[black,dashed]  (3,0+9) -- (1.5,-1.5+9) node[midway, above, sloped]{\small $a$};
    \draw[black]  (1.5,-1.5+9) -- (3,-3+9) node[midway, above, sloped]{\small $\tau_{B_3}$};
    \filldraw [black] (3,-3+9) circle (2pt) node[below]{$\mathsf{V_B}$};
    \filldraw [lightgray] (1.5,-1.5+9) circle (2pt);
    \draw[->, black] (-1.5, -1.5+4.5) -- (-0.75, -1.5+4.5) node[midway, above]{\small $\tau_{A_2B_2}$};
	\draw[->,black,thick] (-0.5,-3+4.5) -- (-0.5,0.5+4.5) node[above]{$t$};
	\draw[->,black,thick] (-0.5,-3+4.5) -- (3.5,-3+4.5) node[below]{$x$};
    \filldraw [lightgray] (1.5,-3+4.5) circle (2pt) node[below]{$\mathsf{P}$};	
    
    \filldraw [black] (0,0+4.5) circle (2pt);
    \draw[black,dashed]  (0,0+4.5) -- (1.5,-1.5+4.5) node[midway, above, sloped]{\small $a$};
    \draw[black]  (1.5,-1.5+4.5) -- (0,-3+4.5) node[midway, above, sloped]{\small $\tau_{A_2}$};
    \filldraw [black] (0,-3+4.5) circle (2pt) node[below]{$\mathsf{V_A}$};
    
    \filldraw [black] (3,0+4.5) circle (2pt);
    \draw[black,dashed]  (3,0+4.5) -- (1.5,-1.5+4.5) node[midway, above, sloped]{\small $a$};
    \draw[black]  (1.5,-1.5+4.5) -- (3,-3+4.5) node[midway, above, sloped]{\small $\tau_{B_2}$};
    \filldraw [black] (3,-3+4.5) circle (2pt) node[below]{$\mathsf{V_B}$};
    \filldraw [lightgray] (1.5,-1.5+4.5) circle (2pt);
    \draw[->, black] (-1.5, -1.5) -- (-0.75, -1.5) node[midway, above]{\small $\tau_{A_1B_1}$};
	\draw[->,black,thick] (-0.5,-3) -- (-0.5,0.5) node[above]{$t$};
	\draw[->,black,thick] (-0.5,-3) -- (3.5,-3) node[below]{$x$};
    \filldraw [lightgray] (1.5,-3) circle (2pt) node[below]{$\mathsf{P}$};	
    
    \filldraw [black] (0,0) circle (2pt);
    \draw[black,dashed]  (0,0) -- (1.5,-1.5) node[midway, above, sloped]{\small $a$};
    \draw[black]  (1.5,-1.5) -- (0,-3) node[midway, above, sloped]{\small $\tau_{A_1}$};
    \filldraw [black] (0,-3) circle (2pt) node[below]{$\mathsf{V_A}$};
    
    \filldraw [black] (3,0) circle (2pt);
    \draw[black,dashed]  (3,0) -- (1.5,-1.5) node[midway, above, sloped]{\small $a$};
    \draw[black]  (1.5,-1.5) -- (3,-3) node[midway, above, sloped]{\small $\tau_{B_1}$};
    \filldraw [black] (3,-3) circle (2pt) node[below]{$\mathsf{V_B}$};
    \filldraw [lightgray] (1.5,-1.5) circle (2pt);
\end{tikzpicture}
}
}
\caption{Visual representation of splitting into two new QPV protocols from the existence of a quantum communication attack on a single QPV protocol. Two attackers $A,B$ receive inputs $\rho_A, \rho_B$ and apply some channel $\mathcal{A}(\rho_{A}) = \sigma_{A_1 A_2}, \mathcal{B}(\rho_{B}) = \sigma_{B_1 B_2}$ and send parts of their outcome to the other party. This procedure defines two new QPV protocols. If there again exists a sufficiently successful quantum communication attack for both new protocols, then by the same argument we can define 4 new QPV protocols, and so on.}
\label{fig:qpvdoubledown}
\end{figure}    

If no branch stopped before depth $n$, these attacks would redistribute each original input over $2^n$ leaf registers while decreasing honest success by less than $n\delta_n$. The Qi--Ranard theorem then gives a common leaf at which both accumulated local channels are approximately measure-and-prepare. LOCC attackers could implement the approximating channels by measuring the original inputs locally and exchanging their classical outcomes, and could then simulate the honest measurement of that leaf. With the choice of $\delta_n$ above, their success would exceed $c_0$, contradicting the initial LOCC gap $\Delta$. There must therefore be at least one stopping node before depth $n$.

\begin{theorem}[Recursive splitting produces a quantum-communication gap]
\label{thm:revised-approximate-splitting}
For every integer $n\geq2$ satisfying condition~\eqref{eq:revised-depth-condition}, and for any choices of attacks satisfying $s_u>p_{\mathrm{hon}}(\mathcal P_u)-\delta_n$ at the nonterminal nodes, the resulting attack-induced tree contains a stopping node at depth strictly less than $n$. In particular, the tree contains a protocol $\mathcal P_v$, obtained from $\mathcal P_{\varnothing}$ by fewer than $n$ attack-induced splitting steps,
such that
\begin{align}
  \omega_{\rm LOQC}(\mathcal P_v) &\leq p_{\mathrm{hon}}(\mathcal P_v)-\delta_n
  \label{eq:revised-loqc-gap}
\end{align}
Here $\delta_n$ is given by~\eqref{eq:revised-delta}. The stopping node may be the root $v=\varnothing$, in which case the initial protocol already has the required gap.
\end{theorem}

The operational meaning of the inequality is that the induced protocol $\mathcal P_v$ is secure against quantum communication. Indeed, since $\omega_{\rm LOQC}(\mathcal P_v)$ is the supremum over all allowed LOQC attacks, no such attack can come within $\delta_n$ of the honest success probability. We started from a protocol for which only security against LOCC attackers was assumed, and the splitting construction therefore yields an induced protocol that is secure against LOQC attackers. Furthermore, the splitting operations that generate this induced protocol can all be implemented locally on the verifiers inputs, showing that the original inputs contained too much information about the problem. Which is apparent in the protocol of  Proposition \ref{prop:strict-qc-advantage}. 

We use the following channel form of the emergent-classicality theorem.

\begin{theorem}[Qi--Ranard~\cite{qi2020emergent}]
\label{thm:QiRanard-revised}
Let $\Lambda:\dens(A)\to\dens(B_1\otimes\cdots\otimes B_m)$ be a quantum channel. For integers $q,r\in\{1,\ldots,m\}$, there exist a POVM $\{M_\alpha\}_\alpha$ on $A$ and an excluded set $Q\subseteq\{1,\ldots,m\}$ with $|Q|\leq q$ such that, for every $R_0\subseteq\{1,\ldots,m\}\setminus Q$ with $|R_0|=r$, the reduced channel $\Lambda_{R_0}:=\Tr_{\overline{R_0}}\circ\Lambda$ satisfies
\begin{equation}
  \lVert\Lambda_{R_0}-\mathcal E_{R_0}\rVert_\diamond
  \leq d_A^3\sqrt{2\ln(d_A)\frac{r}{q}},
  \label{eq:revised-qi-ranard}
\end{equation}
where
\begin{equation}
  \mathcal E_{R_0}(X)
  :=\sum_\alpha\Tr(M_\alpha X)\,\sigma_{R_0}^\alpha
\end{equation}
is a measure-and-prepare channel. The POVM is independent of $R_0$, although the prepared states may depend on $R_0$.
\end{theorem}

\begin{proof}[Proof of Theorem~\ref{thm:revised-approximate-splitting}]
Fix any choices of attacks allowed by the recursive construction, and assume for contradiction that the resulting tree has no stopping node at depth less than $n$. Every node $v$ at depth less than $n$ therefore satisfies
\begin{equation}
  \omega_{\rm LOQC}(\mathcal P_v)
  >p_{\mathrm{hon}}(\mathcal P_v)-\delta_n.
  \label{eq:revised-no-gap}
\end{equation}
At every such node, the chosen attack has success
\begin{equation}
  s_v>p_{\mathrm{hon}}(\mathcal P_v)-\delta_n.
\end{equation}
By construction, these attacks form the full binary splitting tree through depth $n$. Iterating Lemma~\ref{lem:revised-one-splitting-step} gives, for every $v\in\{0,1\}^t$ and $1\leq t\leq n$,
\begin{equation}
  p_{\mathrm{hon}}(\mathcal P_v)>p_{\mathrm{hon},0}-t\delta_n,
  \qquad
  \omega_{\rm LOCC}(\mathcal P_v)\leq c_0.
  \label{eq:revised-tree-induction}
\end{equation}

We now write the channels accumulated in the tree explicitly. Set $\Lambda_0^A=\mathrm{id}_A$ and $\Lambda_0^B=\mathrm{id}_B$. If $\mathcal A_v:A_v\to A_{v0}A_{v1}$ and $\mathcal B_v:B_v\to B_{v0}B_{v1}$ are the local splitting channels at node $v$, define
\begin{align}
  \Lambda_{t+1}^A
  &:=\left(\bigotimes_{v\in\{0,1\}^t}\mathcal A_v\right)
      \circ\Lambda_t^A,
  &
  \Lambda_{t+1}^B
  &:=\left(\bigotimes_{v\in\{0,1\}^t}\mathcal B_v\right)
      \circ\Lambda_t^B.
  \label{eq:revised-tree-channels}
\end{align}
Thus $\Lambda_n^A$ has the $2^n$ leaf registers $A_v$ as outputs, and $\Lambda_n^B$ has the $2^n$ leaf registers $B_v$ as outputs. For a leaf $v\in\{0,1\}^n$, let $\Lambda_v^A$ and $\Lambda_v^B$ denote the corresponding reduced channels.

Apply Theorem~\ref{thm:QiRanard-revised} separately to $\Lambda_n^A$ and $\Lambda_n^B$, with $q=2^{n-1}-1$ and $r=1$. It gives excluded leaf sets $Q_A,Q_B$, each of size at most $2^{n-1}-1$. Since
\begin{equation}
  |Q_A\cup Q_B|\leq2(2^{n-1}-1)=2^n-2<2^n,
\end{equation}
there is a common leaf $i\notin Q_A\cup Q_B$, i.e.~excluded by neither set. For this leaf there are measure-and-prepare channels $\mathcal E_i^A$ and $\mathcal E_i^B$ such that
\begin{equation}
  \lVert\Lambda_i^A-\mathcal E_i^A\rVert_\diamond
  \leq\varepsilon_{A,n},
  \qquad
  \lVert\Lambda_i^B-\mathcal E_i^B\rVert_\diamond
  \leq\varepsilon_{B,n}.
  \label{eq:revised-local-diamond}
\end{equation}

The required joint estimate follows from the telescoping identity
\begin{align}
 \Lambda_i^A\otimes\Lambda_i^B
  -\mathcal E_i^A\otimes\mathcal E_i^B =(\Lambda_i^A-\mathcal E_i^A)\otimes\Lambda_i^B
  +\mathcal E_i^A\otimes(\Lambda_i^B-\mathcal E_i^B).
  \label{eq:revised-telescope}
\end{align}
The diamond norm is multiplicative under tensor products and every quantum channel has diamond norm one~\cite{watrous_theory_2018}. Hence
\begin{equation}
  \lVert\Lambda_i^A\otimes\Lambda_i^B
  -\mathcal E_i^A\otimes\mathcal E_i^B\rVert_\diamond
  \leq\varepsilon_{A,n}+\varepsilon_{B,n}.
  \label{eq:revised-joint-diamond}
\end{equation}
This estimate already includes entangled inputs and arbitrary reference systems. In particular, set
\begin{align}
  \omega_i
  &:=(\mathrm{id}_R\otimes\Lambda_i^A\otimes\Lambda_i^B)(\rho_{RAB}),\\
  \widetilde\omega_i
  &:=(\mathrm{id}_R\otimes\mathcal E_i^A\otimes\mathcal E_i^B)
     (\rho_{RAB}).
\end{align}
Then
\begin{equation}
  \lVert\omega_i-\widetilde\omega_i\rVert_1
  \leq\varepsilon_{A,n}+\varepsilon_{B,n}.
  \label{eq:revised-state-distance}
\end{equation}

Write the measure-and-prepare channels as
\begin{equation}
  \mathcal E_i^A(X)=\sum_\alpha\Tr(M_\alpha^A X)\sigma_{A_i}^\alpha,
  \qquad
  \mathcal E_i^B(Y)=\sum_\beta\Tr(M_\beta^B Y)\tau_{B_i}^\beta.
\end{equation}
For an entangled input, their joint output is generally not a product of averaged marginals. Rather,
\begin{equation}
  \widetilde\omega_i
  =\sum_{\alpha,\beta}\rho_R^{\alpha\beta}\otimes
     \sigma_{A_i}^\alpha\otimes\tau_{B_i}^\beta,
  \label{eq:revised-correlated-separable}
\end{equation}
where
\begin{equation}
  \rho_R^{\alpha\beta}
  :=\Tr_{AB}\!\left[
    (\mathbbm 1_R\otimes M_\alpha^A\otimes M_\beta^B)\rho_{RAB}
  \right]
\end{equation}
is an unnormalised conditional reference state. This correlated separable state is implementable by one-round LOCC: the attackers measure locally, exchange $\alpha$ and $\beta$, and then both know the conditional product state $\sigma_{A_i}^\alpha\otimes\tau_{B_i}^\beta$. If $\{\Pi_z^i\}_z$ is the public honest POVM of leaf protocol $\mathcal P_i$, they use shared randomness to sample
\begin{equation}
  q(z\mid\alpha,\beta)
  :=\Tr\!\left[\Pi_z^i
     (\sigma_{A_i}^\alpha\otimes\tau_{B_i}^\beta)\right]
  \label{eq:revised-sampling}
\end{equation}
and return the same sampled answer. This exactly simulates the honest leaf measurement on $\widetilde\omega_i$.

Let $\Pi_i^{\mathrm{hon}}$ denote the overall acceptance POVM element induced by the specified honest-prover measurement for $\mathcal P_i$ and the verifiers' acceptance test. The same measurement applied to $\omega_i$ and $\widetilde{\omega}_i$ gives $p_{\mathrm{hon}}(\mathcal P_i)$ and $p^{\mathrm{LOCC}}_{\mathrm{sim},i}$, respectively. Therefore,
\begin{align}
\left|p_{\mathrm{hon}}(\mathcal P_i)-p^{\mathrm{LOCC}}_{\mathrm{sim},i}\right|
&=\left|\Tr\!\left[\Pi_i^{\mathrm{hon}}\omega_i\right]-\Tr\!\left[\Pi_i^{\mathrm{hon}}\widetilde{\omega}_i\right]\right| \nn\\
&=\left|\Tr\!\left[\Pi_i^{\mathrm{hon}}\left(\omega_i-\widetilde{\omega}_i\right)\right]\right| \nn\\
&\leq\frac12\left\|\omega_i-\widetilde{\omega}_i\right\|_1 \nn\\
&\leq\frac{\varepsilon_{A,n}+\varepsilon_{B,n}}{2}.
\end{align}

Consequently,
\begin{align}
  c_0
  &\geq p^{\mathrm{LOCC}}_{\mathrm{sim},i}\nn\\
  &\geq p_{\mathrm{hon}}(\mathcal P_i)
    -\frac{\varepsilon_{A,n}+\varepsilon_{B,n}}{2}\nn\\
  &>p_{\mathrm{hon},0}-n\delta_n
    -\frac{\varepsilon_{A,n}+\varepsilon_{B,n}}{2}\nn\\
  &=c_0+\frac12\left(
      \Delta-\frac{\varepsilon_{A,n}+\varepsilon_{B,n}}{2}
    \right)>c_0,
\end{align}
so we get a contradiction. Thus some node at depth less than $n$ satisfies \eqref{eq:revised-loqc-gap}. 
\end{proof}


\subsection{Security of \texorpdfstring{QPV$_{\textsf{Bell}}$}{QPVbell} against quantum communication} \label{sec:QCproofagainstBellstatediscrimination}

In this section we give the first example of a classically loss-tolerant QPV protocol that is secure against attackers restricted to quantum communication. Furthermore we will show that there is no perfect attack with loss in the quantum communication setting for this protocol, making it the first example of a protocol that is secure against lossy quantum communication attacks with no pre-shared entanglement. We will give two different proofs that both show security against quantum communication, one based on monogamy of entanglement and one based off a no-cloning argument. The monogamy of entanglement gives a worse bound, but seems to be a more protocol-agnostic approach to proving security in the quantum communication setting, making its derivation still useful.

The protocol we look into is the Bell state discrimination problem. Two verifiers send as inputs the respective qubits of one of the four Bell states ($\ket{\Phi^+}, \ket{\Phi^-}, \ket{\Psi^+}, \ket{\Psi^-}$) and ask the prover which Bell state he receives. An honest prover can answer this task perfectly by doing a Bell measurement. With an SDP and a similar analysis as in \cite{allerstorfer2022practical}, we can show that this protocol is secure and loss-tolerant against attackers restricted to classical communication. An optimal attacking strategy turns out to measure the qubits in the computational basis, which distinguishes $\{ \ket{\Phi^+}, \ket{\Phi^-} \}$ from $\{ \ket{\Psi^+}, \ket{\Psi^-} \}$ and then to guess one of the two Bell states as an answer. This has success probability $1/2$ and is optimal. 

To analyze security against attackers restricted to quantum communication we look at the protocol in the following equivalent purified way.
\begin{itemize}
    \item The inputs of the protocol will be the second qubit of the maximally entangled state $\frac{1}{\sqrt{2}} (\ket{00} + \ket{11})_{V_AA}$ and $\frac{1}{\sqrt{2}} (\ket{00} + \ket{11})_{V_BB}$. 
    \item The honest prover has to do a Bell measurement, this is now an entanglement swapping operation where the entanglement between the verifiers and the prover gets swapped to entanglement between the two verifiers and entanglement between the two qubits the prover holds. The prover then sends out his classical measurement outcome $i$ to the verifiers. 
    \item The verifiers check whether the entanglement swapping operation was successful by applying a Bell measurement on their joint state and check whether their measurement outcome is the same as the answer of the prover. The probability of successfully attacking the protocol now corresponds to the provers having the correct Bell state as measurement outcome, averaged over all possible Bell states, i.e. $p_{\text{succ}} = \frac{1}{4} \sum_i \Tr[\ket{\text{Bell}_i}\bra{\text{Bell}_i} \rho_{V_AV_B}^i]$. 
\end{itemize}
Note that from the point of view of the prover nothing changes from the original Bell state discrimination protocol. The honest prover still needs to perform a Bell state measurement on his incoming qubits and send his measurement outcome to both verifiers. 

The idea of this proof is that by reformulating the QPV$_{\textsf{Bell}}$ protocol as an entanglement swapping protocol we can use the monogamy of entanglement property between the qubits that remain at the verifiers and the quantum systems attackers create. Furthermore, while it is hard to say anything about the quantum systems attackers might send to each other, the states the verifiers keep are always in their control. 
    
As stated in section \ref{sec:splitting} the most general quantum communication attack is for attackers to split their inputs into two quantum systems. They hold on to one and forward the other system to the other attacker. After the quantum communication round the attackers hold the reduced states $\rho_{A_1B_1}$ and $\rho_{A_2B_2}$, respectively. Since there is no more further communication, for a quantum attack to be successful in generating entanglement between the two verifiers, it is sufficient to look at only one of the two attackers locally. Attackers have to act in a coordinated way in QPV, the measurement outcome of the second attacker has to be equal to first. Therefore the second measurement will not disturb the overall state at all by the gentle measurement lemma \cite{winter1999coding}, thus the first measurement will already swap the entanglement to between the verifiers. 

\begin{figure}[ht]
\centering
\begin{subfigure}{0.45\textwidth}
    \centering
    \begin{tikzpicture}
        \filldraw [black] (0,0) circle (2pt) node[left] {$\mathsf{V_A}$};
        \draw[black] (0.707,0.707) circle (0pt) node[left=1mm] {$\scriptstyle{e_{A_1} \leq 1/2}$};
        \draw[black] (0.707,-0.707) circle (0pt) node[left=1mm] {$\scriptstyle{e_{A_2}}$};
        \draw[black] [,decorate,decoration=snake] (0,0) -- (1.25,1);
        \draw[black] [,decorate,decoration=snake] (0,0) -- (1.25,-1);
        \filldraw [black] (1.25,1) circle (2pt) node[above=1.5mm] {$\mathsf{A_1}$};
        
        \draw[lightgray] [,decorate,decoration=snake] (1.25,1) -- (1.25,-1);
        \draw[lightgray] [,decorate,decoration=snake] (1.75,1) -- (1.75,-1);
    
        \filldraw [black] (3,0) circle (2pt) node[right] {$\mathsf{V_B}$};
        \draw[black] (3-0.707,0.707) circle (0pt) node[right=1mm] {$\scriptstyle{e_{B_1}}$};
        \draw[black] (3-0.707,-0.707) circle (0pt) node[right=1mm] {$\scriptstyle{e_{B_2}}$};
        \draw[black] [,decorate,decoration=snake] (3,0) -- (1.75,1);
        \draw[black] [,decorate,decoration=snake] (3,0) -- (1.75,-1);
        \filldraw [black] (1.75,1) circle (2pt) node[above=1.5mm] {$\mathsf{B_1}$};
    
        \filldraw [black] (1.25,-1) circle (2pt) node[below=1.5mm] {$\mathsf{A_2}$};
        \filldraw [black] (1.75,-1) circle (2pt) node[below=1.5mm] {$\mathsf{B_2}$};
        
        \draw[QuSoft, thick] (1,-0.8) rectangle (2,-1.2);
        \draw[QuSoft, thick] (1,0.8) rectangle (2,1.2);
    \end{tikzpicture}
\caption{Entanglement structure when attackers measure and commit to an answer. W.l.o.g. $e_{A_1} \leq 1/2$.}
\label{fig:Entangl_AB_meas}
\end{subfigure}
\begin{subfigure}{0.45\textwidth}
    \centering
    \begin{tikzpicture}
        \filldraw [black] (0,-1) circle (2pt) node[left] {$\mathsf{V_A}$};
        \draw[black] [,decorate,decoration=snake] (0,-1) -- (3,-1);
        \filldraw [black] (1.25,0) circle (2pt) node[above=1.5mm] {$\mathsf{A_1}$};
    
        \filldraw [black] (3,-1) circle (2pt) node[right] {$\mathsf{V_B}$};
        \draw[black] [,decorate,decoration=snake] (1.25,0) -- (1.75,0);
        \filldraw [black] (1.75,0) circle (2pt) node[above=1.5mm] {$\mathsf{B_1}$};
        
        \draw[black] (1.5,-1) circle (0pt) node[above=1mm] {$\scriptstyle{e_{V_AV_B} \leq 1/2}$};
        
        \draw[black] (1.5,1) circle (0pt) node[above=2mm] {};
        \draw[white, thick] (1,-1.5) rectangle (2,-2);
    
    \end{tikzpicture}
    \caption{Measuring $\mathsf{A}_1 \mathsf{B}_1$ and sending the result to $\mathsf{V}_\mathsf{A}$ is an LOCC operation on $V_\mathsf{A}(\mathsf{A}_1 \mathsf{B}_1V_\mathsf{B})$.}
    \label{fig:Entangl_VAVB}
\end{subfigure}
\caption{Illustration of the argument based on monogamy of entanglement to bound the entanglement $e_{V_AV_B} \leq 1/2$ as described in the main text. Tracing out $\mathsf{A_2B_2}$, attacker $\mathsf{A}$ can only generate $e_{V_AV_B} \leq 1/2$ ebits between $\mathsf{V}_\mathsf{A}$ and $\mathsf{V}_\mathsf{B}$.}
\label{fig:QCmonogamy_arg}
\end{figure}
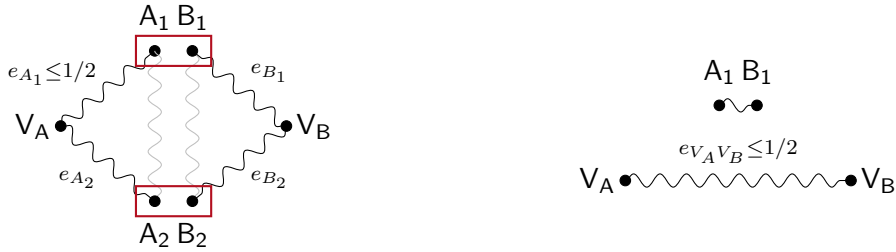

 As an entanglement measure we will use the squashed entanglement. This measure satisfies several properties useful for our analysis, such as monotonicity under LOCC operations, general monogamy with no restrictions on the size of the quantum registers, and it is lower bounded by distillable entanglement \cite{christandl2004squashed}. Consider the sketch of the entanglement structure between all quantum registers in figure \ref{fig:QCmonogamy_arg}. By monogamy we have that 
\begin{align}
    0 \leq E_{sq}(V_A:A_1)_{\rho} + E_{sq}(V_A:A_2)_{\rho} \leq E_{sq}(V_A:A)_{\rho} = E_{sq}(\ket{\Phi^+} \bra{\Phi^+}) = 1
\end{align}
Suppose without loss of generality that $E_{sq}(V_A:A_1)_{\rho} \leq E_{sq}(V_A:A_2)_{\rho}$, then the inequality implies that $E_{sq}(V_A:A_1)_{\rho} \leq 1/2$. Let $\Phi$ be the LOCC operation (on $\mathsf{V_A (A_1 B_1 V_B)}$) of measuring the $\rho_{A_1B_1}$ register and sending the classical measurement result to $\mathsf{V_A}$. Using that squashed entanglement is monotone under LOCC we get the following 
\begin{align}
    E_{sq}(V_A : V_B)_{\Phi(\rho)} &\leq E_{sq}(V_A:A_1B_1V_B)_{\Phi(\rho)} \nn \\
    &\leq E_{sq}(V_A:A_1B_1V_B)_{\rho} \nn \\
    &= E_{sq}(V_A:A_1) \leq 1/2. \label{squashedboundhalf}
\end{align}
Thus, the squashed entanglement between the two verifiers after any attack using quantum communication is upper bounded by 1/2. Recall that for an attack to be successful, the verifiers must share the same Bell state on their registers $\rho_{V_AV_B}$ as the answer $i$ they receive from the attackers. Since $E_{sq}(\ket{B_i}\bra{B_i}) = 1$ it is immediately clear that $p_{\text{succ}} = \sum_i \Tr[\ket{B_i}\bra{B_i} \rho_{V_AV_B}^i]/4 < 1$ and no perfect attack is possible. 

Ideally we want $p_\text{succ}$ to not only be strictly smaller than $1$, but to be smaller than $1$ by some finite gap. In what follows we will show that $E_{sq}(V_A:V_B) \leq 1/2$ implies $p_{\text{succ}} \leq 0.926$, which implies security against quantum communication. Our proof uses the \textit{hashing bound} from \cite{devetak2005distillation} which lower bounds the squashed entanglement \cite{christandl2004squashed}. The inequality states that for any quantum state $\rho_{AB}$,
\begin{align} \label{eq:hasingbound}
    S(B)_\rho - S(AB)_\rho \leq E_{sq}(A:B)_\rho,
\end{align}
where $\rho_B = \Tr_A[\rho_{AB}]$, and $S$ is the von Neumann entropy.

The idea of this proof is to apply the \textit{Werner twirling channel} $\mathcal{W}$, where we integrate over the final two-qubit state between the verifiers. This channel leaves the antisymmetric (qubit) state invariant, and projects the remaining part to the symmetric subspace. Furthermore this channel is an LOCC channel and by monotonicity of the squashed entanglement under LOCC operations, we have
\begin{align} \label{eq:sqwernerentanglementbound}
    1/2 \geq E_{sq}(V_A:V_B)_{\Phi(\rho)} \geq E_{sq}(V_A:V_B)_{\mathcal{W}(\Phi(\rho))}. 
\end{align}
The resulting state $\mathcal{W}(\Phi(\rho)_{V_AV_B})$ can then be written as a mixture of the antisymmetric Bell state with the maximally mixed state characterized by some $\alpha > 0$, that is, 
\begin{align}
    \mathcal{W}(\Phi(\rho)_{V_AV_B}) = \alpha \ketbra{\Psi^-}{\Psi^-}+ (1-\alpha) \frac{\mathbbm{1}_4}{4}.   
\end{align}
A property of Bell states is that they can be locally transformed into one another. Therefore, verifiers can always locally change the Bell state that they receive as an answer from the honest prover to the antisymmetric Bell state. Therefore any successful attack can be characterized by the probability of having measurement outcome $\ket{\Psi^-}$ on $\mathcal{W}(\Phi(\rho_{V_AV_B}))$. Combining the entanglement bound \eqref{eq:sqwernerentanglementbound} with the hashing bound \eqref{eq:hasingbound} we get the following numerical bound on $\alpha$:
\begin{align}
    1/2 &\geq E_{sq}(V_A:V_B)_{\mathcal{W}(\Phi(\rho))} \nn \\
    &\geq S(B)_{\mathcal{W}(\Phi(\rho))} - S(AB)_{\mathcal{W}(\Phi(\rho))} \nn \\
    &= 1 - S(AB)_{\mathcal{W}(\Phi(\rho))} \nn \\
    \Rightarrow \ \ \ \ \ \alpha &\leq 0.901.
\end{align}
The probability of success is now upper bounded as follows
\begin{align}\label{equ:QCupper}
    p_{\text{succ}} &= \Tr\left[\ketbra{\Psi^-}{\Psi^-} \Phi(\rho)_{V_AV_B}\right] = \mathcal{F}(\ketbra{\Psi^-}{\Psi^-}, \Phi(\rho)_{V_AV_B}) \leq \mathcal{F}(\mathcal{W}(\ketbra{\Psi^-}{\Psi^-}), \mathcal{W}(\Phi(\rho)_{V_AV_B})) \nn \\
    &= \Tr\left[\ketbra{\Psi^-}{\Psi^-} \mathcal{W}(\Phi(\rho)_{V_AV_B})\right] = \alpha + \frac{1-\alpha}{4} \leq 0.926,
\end{align}
where we have used the data process inequality for the fidelity in the first inequality. This concludes our proof and shows there is a finite gap between the optimal attack attackers restricted to quantum communication can do and what an honest prover can do. We suspect that this gap can be made even larger, the upper bound that we find arises only due to restrictions on the $A_1$ part, the $B_1$ part could be unchanged from the input state of $B$. So our bound gives an expression for the maximal probability if you split the $A$ part in two parts but get the full $B$ part. Also we have not yet made use of the fact that both attackers have to answer equally. In the following section we try to improve the bound using a different approach making use of both the $A_1$ and the $A_2$ part.

\subsubsection{An improved bound via no-cloning}\label{sec:QCproofagainstBellstatediscrimination_nocloning}
In this section we will show that the optimal approximate cloning bound is strong enough to give a bound that improves \eqref{equ:QCupper} in several ways. This argument gives a bound for any local dimension $d$ of the input, has a clear operational interpretation and improves the above derived bound for $d=2$. For the start, assume that the attackers \textsf{A}, \textsf{B} have optimal local success probabilities $p_\text{succ}^\mathsf{A}, p_\text{succ}^\mathsf{B}$ respectively. Our argument will be that if these local success probabilities (and thus also the global attack) were too high, we would be able to construct a cloning procedure that violates the optimal approximate cloning bound. 

In a realistic attack on QPV$_{\textsf{Bell}}$ it is a requirement that the answers of \textsf{A} and \textsf{B} are identical in each round. In terms of success probabilities this means $p_\text{succ}^\mathsf{A}= p_\text{succ}^\mathsf{B} = p_\text{succ}^\mathsf{AB}$, the latter denoting the optimal realistic attack. Considering just the local success probabilities is a relaxation because it in principle allows for different responses from \textsf{A} and \textsf{B}. It is clear that 
\begin{align}
    p_\text{succ}^\mathsf{AB} \leq \min \left\{ p_\text{succ}^\mathsf{A}, p_\text{succ}^\mathsf{B} \right\},
\end{align}
and we will now proceed to upper bound $\min \left\{ p_\text{succ}^\mathsf{A}, p_\text{succ}^\mathsf{B} \right\}$. As previously mentioned, the structure of any attack (described at the start of section \ref{sec:splitting}) on QPV$_{\textsf{Bell}}$ must lead to an entanglement structure as drawn in Figure~\ref{fig:Entangl_AB_meas}. In particular, the verifiers themselves (or even just one of them) could create the situation of Figure~\ref{fig:Entangl_AB_meas} themselves in their own lab because they know what the optimal success strategies are. They could simply create the inputs and apply the optimal strategies of \textsf{A}, \textsf{B} to registers $A, B$. This splits $A \mapsto A_1 A_2$ and $B \mapsto B_1 B_2$. Imagine now that we add a third input state $\ket{\Phi_+}_{V_CC}$ in registers $V_CC$ and that the optimal strategy of \textsf{B} is applied to register $C$, mapping $C \mapsto C_1 C_2$. Note that the state in $C_1C_2$ is identical to the one in $B_1B_2$. This procedure creates the following entanglement structure:

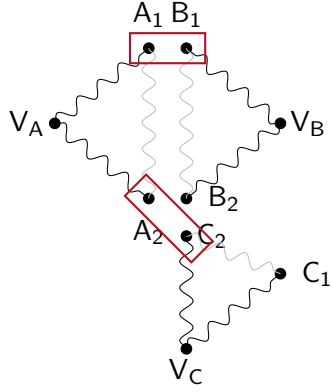
\begin{figure}[ht]
\centering
    \begin{tikzpicture}
        \filldraw [black] (0,0) circle (2pt) node[left] {$\mathsf{V_A}$};
        \draw[black] [,decorate,decoration=snake] (0,0) -- (1.25,1);
        \draw[black] [,decorate,decoration=snake] (0,0) -- (1.25,-1);
        \filldraw [black] (1.25,1) circle (2pt) node[above=1.5mm] {$\mathsf{A_1}$};
    
        \filldraw [black] (3,0) circle (2pt) node[right] {$\mathsf{V_B}$};
        \draw[black] [,decorate,decoration=snake] (3,0) -- (1.75,1);
        \draw[black] [,decorate,decoration=snake] (3,0) -- (1.75,-1);
        \filldraw [black] (1.75,1) circle (2pt) node[above=1.5mm] {$\mathsf{B_1}$};
        
        \filldraw [black] (1.75,-3) circle (2pt) node[below] {$\mathsf{V_C}$};
        \draw[black] [,decorate,decoration=snake] (1.75,-3) -- (1.75,-1.5);
        \draw[black] [,decorate,decoration=snake] (1.75,-3) -- (3,-2);
        \filldraw [black] (1.75,-1.5) circle (2pt) node[right] {$\mathsf{C_2}$};
    
        \filldraw [black] (1.25,-1) circle (2pt) node[below=1.5mm] {$\mathsf{A_2}$};
        \filldraw [black] (1.75,-1) circle (2pt) node[right=1.5mm] {$\mathsf{B_2}$};
        \filldraw [black] (3,-2) circle (2pt) node[right=1.5mm] {$\mathsf{C_1}$};
        
        \draw[lightgray] [,decorate,decoration=snake] (1.25,1) -- (1.25,-1);
        \draw[lightgray] [,decorate,decoration=snake] (1.75,1) -- (1.75,-1);
        \draw[lightgray] [,decorate,decoration=snake] (1.75,-1.5) -- (3,-2);

        \draw[QuSoft, thick] (1,0.8) rectangle (2,1.2);
        \draw[QuSoft, thick, rotate around={-45:(1.25,-1)}] (1,-0.8) rectangle (2.25,-1.2);
    \end{tikzpicture}
    \caption{Entanglement structure created by someone, say a trusted verifier, who generated three inputs and applied the optimal local attacker channel $\mathcal{A}$ of \textsf{A} on one and the analogous channel $\mathcal{B}$ of \textsf{B} on two of the three inputs. They then apply the optimal measurement of \textsf{A} on registers $A_1B_1$ and the one of \textsf{B} on registers $A_2C_2$. }
\label{fig:qc_upper_nocloning2}
\end{figure}

Afterwards, the optimal measurement of \textsf{A} is applied to registers $A_1B_1$ and the optimal measurement of \textsf{B} is applied to registers $A_2C_2$. This swaps entanglement to the registers $V_AV_B$ and $V_AV_C$, respectively. Tracing all other registers away, we end up with the structure depicted in Figure~\ref{fig:qc_upper_nocloning}.

\begin{figure}[h]
\centering
    \begin{tikzpicture}
        \filldraw [black] (0,0) circle (2pt) node[left] {$\mathsf{V_A}$};
        \draw[black] [,decorate,decoration=snake] (0,0) -- (2,1);
        \draw[black] [,decorate,decoration=snake] (0,0) -- (2,-1);
        
        \draw[black] (1,0.707) circle (0pt) node[above] {$\scriptstyle{\rho_{V_AV_B}}$};
        \draw[black] (1,-0.707) circle (0pt) node[below] {$\scriptstyle{\rho_{V_AV_C}}$};
    
        \filldraw [black] (2,1) circle (2pt) node[right] {$\mathsf{V_B}$};
        
        \filldraw [black] (2,-1) circle (2pt) node[right] {$\mathsf{V_C}$};
    
    \end{tikzpicture}
    \caption{Using the optimal attack strategies of \textsf{A} and \textsf{B} we may create this situation.}
\label{fig:qc_upper_nocloning}
\end{figure}
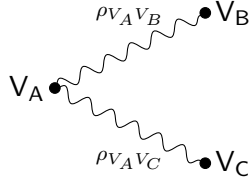

The state $\rho_{V_AV_B}$ corresponds to the state giving \textsf{A} her optimal local success probability in attacking QPV$_{\textsf{Bell}}$. This is because we applied the optimal split $A \mapsto A_1A_2$ and the optimal local measurement on $A_1B_1$ in order to remotely prepare $\rho_{V_AV_B}$. Likewise, the state $\rho_{V_AV_C}$ corresponds to the state giving \textsf{B} her optimal local success probability in attacking QPV$_{\textsf{Bell}}$. At this stage $\mathsf{V_A}$ could take some state $\ket{\phi}_{V_A'}$ and attempt to use the standard teleportation protocol to teleport it to \textit{both} $\mathsf{V_B}$ and $\mathsf{V_C}$. In general, this will result in some states $\rho_B^{\phi}$ and $\rho_C^{\phi}$ at $\mathsf{V_B}$ and $\mathsf{V_C}$, respectively. The average teleportation fidelity $f$ between the resultant state and the original one depends only on the maximally entangled fraction of the resource state $\rho$ \cite{HorodeckiImnperfectTeleportation}, and is given by
\begin{align}\label{equ:imperf_cloning}
    f = \frac{Fd+1}{d+1} = \frac{\braket{\Phi_+|\rho|\Phi_+}d+1}{d+1},
\end{align}
for any local dimension $d$. A consequence of optimal asymmetric $1 \to 2$ cloning is that the arithmetic mean of the average fidelities fulfills
\begin{align}\label{equ:avg_fidelities}
    \frac{f_\mathsf{V_B} + f_\mathsf{V_C}}{2} \leq \frac{5}{6},
\end{align}
no matter which resource states was used. Plugging in equation \eqref{equ:imperf_cloning} for each $f$ in \eqref{equ:avg_fidelities} yields
\begin{align}\label{equ:psuccavg}
    \frac{\braket{\Phi_+|\rho_{V_AV_B}|\Phi_+} + \braket{\Phi_+|\rho_{V_AV_C}|\Phi_+}}{2} \leq \frac{5}{6} - \frac{1}{6d}.
\end{align}
This allows us to bound the average success probability $(p_\text{succ}^\mathsf{A} + p_\text{succ}^\mathsf{B})/2$ as follows. Note that any Bell state $\ket{B_i}$ can be regarded as $\ket{\Phi_+}$ with a suitable local unitary $\mathds{1} \otimes U_i$ applied to the latter. We can thus write
\begin{align}
    \frac{p_\text{succ}^\mathsf{A} + p_\text{succ}^\mathsf{B}}{2} &= \frac{1}{d^2} \sum_{i=0}^{d^2-1} \frac{\braket{\Phi_+ | U_i^\dagger \rho_{V_AV_B}^{i} U_i | \Phi_+} + \braket{\Phi_+ | U_i^\dagger \rho_{V_AV_C}^{i} U_i | \Phi_+}}{2} \\
    &\leq \frac{1}{d^2} \sum_{i=0}^{d^2-1} \left( \frac{5}{6} - \frac{1}{6d} \right) \\
    &=\frac{5}{6} - \frac{1}{6d},
\end{align}
where the inequality follows from the fact that \eqref{equ:psuccavg} holds for \textit{any} states $\rho_{V_AV_B}$ and $\rho_{V_AV_C}$. Having an upper bound on the average success probability of \textsf{A} and \textsf{B} also tells us that at least one of $p_\textsf{succ}^\mathsf{A}$ and $p_\textsf{succ}^\mathsf{B}$ is upper bounded by the same bound. Hence we arrive at our end result
\begin{align}\label{equ:final_psucc_bound}
    p_\text{succ}^{\mathsf{AB}}(d) \leq \min \left\{ p_\text{succ}^\mathsf{A}(d), p_\text{succ}^\mathsf{B}(d) \right\} \leq \frac{p_\text{succ}^\mathsf{A}(d) + p_\text{succ}^\mathsf{B}(d)}{2} \leq \frac{5}{6} - \frac{1}{6d}.
\end{align}
We see that for $d=2$ the previously derived bound improves to $p_\text{succ}^{\mathsf{AB}}(2) \leq 3/4$. In other words, if one tries to discriminate Bell states of local dimension $d$ with only local operations and one round of simultaneous quantum communication (without pre-shared entanglement), the success probability is upper bounded as in \eqref{equ:final_psucc_bound}.

We still suspect that this bound is not tight because we relaxed the situation of QPV and only looked at the local success probabilities of \textsf{A} and \textsf{B}, and their average. In reality, however, \textsf{A} and \textsf{B} are forced to respond with the same answer. Our argument does not include this coordination\footnote{Our argument includes that both \textsf{A} and \textsf{B} have to answer, but their responses may differ.} and therefore we think that the realistic success probability $p_\text{succ}^{\mathsf{AB}}$ is only loosely upper bounded by \eqref{equ:final_psucc_bound}.

\subsection{Lossy Quantum Communication Attack on \texorpdfstring{QPV$_{\textsf{Bell}}$}{QPVbell}}
\label{sec:LossyQuantumCommunicationattackbound}

In the previous section, we considered attacks in which the attackers were allowed to use quantum communication but were not allowed to answer loss. We now extend the attack model by allowing the attackers to return an additional loss symbol. We use the following convention.

\begin{definition}[Lossy attacks on $\mathrm{QPV}_{\textsf{Bell}}$]
\label{def:lossy-qpv-bell}
Let $I\in\{0,1,2,3\}$ denote the correct Bell-state label and let $\varnothing$ denote loss. In every round, the attackers return $Z_A,Z_B\in\{0,1,2,3,\varnothing\}$. The round is discarded if $Z_A=Z_B=\varnothing$ and is retained otherwise. A retained round is successful exactly when $Z_A=Z_B=I$. Thus, unequal answers and one-sided loss reports count as errors.

For an attack strategy $S$, we define its transmission rate and conditional success probability by
\begin{equation}
    \eta(S):=\Pr_S\bigl[(Z_A,Z_B)\neq(\varnothing,\varnothing)\bigr],\qquad p_{\rm succ}(S):=\Pr_S\bigl[Z_A=Z_B=I\mid (Z_A,Z_B)\neq(\varnothing,\varnothing)\bigr].
\end{equation}
We call an attack perfect if $\eta(S)>0$ and $p_{\rm succ}(S)=1$. We call $\mathrm{QPV}_{\textsf{Bell}}$ fully loss tolerant against a given class of attackers if that class contains no perfect attack.
\end{definition}

The rounds for which the answer is `signal loss' ($\perp$) are getting discarded and only the remaining conclusive ones are used for security analysis. For attackers, this provides a new avenue for attacks by claiming loss in rounds when a possible attack has not worked. A protocol is \textit{fully loss tolerant} if it cannot be perfectly attacked for any transmission value $0 < \eta \leq 1$ of the input signals.

Allowing attackers to also answer loss would be the most general setting for attackers who can't pre-share entanglement. Unfortunately our previous proof does not hold in the lossy case, since the operations attackers can apply to their local quantum registers after quantum communication are not LOCC operations (between the local attacker registers and the corresponding verifier) if attackers are allowed to postselect, but rather SLOCC. Thus, the monotonicity of the squashed entanglement does not hold anymore. Similarly, as already pointed out in \cite{HorodeckiImnperfectTeleportation}, average teleportation fidelities may be higher than the bound \eqref{equ:imperf_cloning} if one allows post-selection. However, we still prove that there can be no perfect lossy quantum communication attack on the Bell state discrimination protocol, i.e. $p_{\text{succ}}(\eta)<1$ for any transmission rate $\eta \in (0,1]$.

\begin{figure}[ht]
\centering
    \begin{tikzpicture}
        \filldraw [black] (0,0) circle (2pt) node[left] {$\mathsf{V_A}$};
        \draw[black] [,decorate,decoration=snake] (0,0) -- (1.25,1);
        \draw[black] [,decorate,decoration=snake] (0,0) -- (1.25,-1);
        \filldraw [black] (1.25,1) circle (2pt) node[above=1.5mm] {$\mathsf{A_1}$};
    
        \filldraw [black] (3,0) circle (2pt) node[right] {$\mathsf{V_B}$};
        \draw[black] [,decorate,decoration=snake] (3,0) -- (1.75,1);
        \draw[black] [,decorate,decoration=snake] (3,0) -- (1.75,-1);
        \filldraw [black] (1.75,1) circle (2pt) node[above=1.5mm] {$\mathsf{B_1}$};
        
        \filldraw [black] (1.75,-3) circle (2pt) node[below] {$\mathsf{V_C}$};
        \draw[black] [,decorate,decoration=snake] (1.75,-3) -- (1.75,-1.5);
        \draw[black] [,decorate,decoration=snake] (1.75,-3) -- (3,-2);
        \filldraw [black] (1.75,-1.5) circle (2pt) node[right] {$\mathsf{C_2}$};
    
        \filldraw [black] (1.25,-1) circle (2pt) node[below=1.5mm] {$\mathsf{A_2}$};
        \filldraw [black] (1.75,-1) circle (2pt) node[right=1.5mm] {$\mathsf{B_2}$};
        \filldraw [black] (3,-2) circle (2pt) node[right=1.5mm] {$\mathsf{C_1}$};
        
        \draw[lightgray] [,decorate,decoration=snake] (1.25,1) -- (1.25,-1);
        \draw[lightgray] [,decorate,decoration=snake] (1.75,1) -- (1.75,-1);

        \draw[QuSoft, thick] (1,0.8) rectangle (2,1.2);
        \draw[QuSoft, thick, opacity=0.33, rotate around={-45:(1.25,-1)}] (1,-0.8) rectangle (2.25,-1.2);
        \draw[QuSoft, thick, opacity=0.33] (1,-0.8) rectangle (2,-1.2);
    \end{tikzpicture}
\caption{Entanglement structure including a third hypothetical verifier $\mathsf{V_C}$ and attacker $\mathsf{C}$ who applies an isometry $W_{C \to C_1C_2}$ to his hypothetical input.}
\label{fig:qc_no_lossy_attack}
\end{figure}
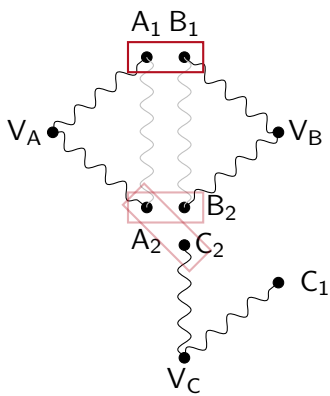

 The argument goes as follows. Suppose there is a perfect lossy quantum communication attack, then for any loss rate there must be some moment where neither attacker declares `signal loss'. Condition on this event taking place, and consider the moment before both attackers measure their quantum systems. Suppose we now perform the measurement on the $A_1B_1$ quantum system, then by assumption we must get some Bell state $i$ (indicating the correct one) as a measurement outcome with probability $1$, thus generating a maximally entangled state between the verifiers $V_A, V_B$. Now consider the possibility of another verifier $V_C$ who also sends as an input a half of an EPR pair, and some attacker $C$ who applies the exact same splitting operation on her input as attacker $B$ (see figure \ref{fig:qc_no_lossy_attack}), thus locally $\rho_{B_2} = \rho_{C_2}$. By definition of a QPV protocol the measurement outcome we get from $A_2,B_2$ will be $i$ with probability $1$. 

What would now happen if we apply the measurement of the attackers on the $A_2, C_2$ system? If this measurement is conclusive then it must be correct with probability 1, thus creating a maximally entangled state between $V_A$ and $V_C$, which would violate monogamy of entanglement since $V_A$ is already maximally entangled with $V_B$. We conclude that the only possible measurement outcome on the $A_2, C_2$ system is an inconclusive `$\varnothing$' outcome. 

Crucially, the measurement on $\rho_{A_2B_2}$ will always be conclusive while the measurement on $\rho_{A_2C_2} = \rho_{A_2} \otimes \rho_{C_2} = \rho_{A_2} \otimes \rho_{B_2}$ will always be inconclusive. Using this fact we can perfectly distinguish the state $\rho_{A_2B_2}$ from  $\rho_{A_2} \otimes \rho_{B_2}$, by simply applying the measurement an attacker would apply. However, the states $\rho_{A_2B_2}$ and $\rho_{A_2} \otimes \rho_{B_2}$ are never orthogonal to each other, thus there can be no procedure that perfectly discriminates the two quantum states without saying loss. This contradicts our findings, because the above is a hypothetical procedure that perfectly distinguishes the two quantum states. 

We conclude that our assumption of a perfect lossy quantum communication attack must be wrong, which proves our claim. Thus we see that in general the QPV$_{\textsf{Bell}}$ cannot be perfectly attacked by unentangled attackers, which is the strongest statement we can make if we require attackers to not pre-share any entanglement. This argument cannot be extended to a finite gap away from 1 in the probability of success of an attack because of the subtlety that the measurements on $\rho_{A_1B_1}$ and $\rho_{A_2C_2}$ can be correlated. If attackers are allowed to make some errors, then the measurements on these states can be correlated such that they can be both conclusive but the measurement on $\rho_{A_2C_2}$ will then just be wrong, so monogamy of entanglement may not be violated.

\subsection{Considerations on loss tolerance in QPV}
\label{sec:losstolerance}
Ideally, for QPV to become feasible, one would like to have a protocol that is fully loss tolerant, secure against attackers being able to pre-share a bounded amount of entanglement and to use quantum communication between them. So far, there is no such protocol.

Here we give a no-go result, based on a simple observation. We show that no such protocol, fulfilling all three of the above properties, can exist. However, not all is lost as in practice one may be able to achieve \textit{good enough} partial loss tolerance, for example by increasing the quantum input dimension or the number of possible quantum operations \textsf{P} can apply. For simplicity, consider the following quite general two-verifier QPV protocol\footnote{The result that follows can be straightforwardly generalized to $m$ verifiers, for which a general attack would be to teleport all quantum inputs to one fixed attacker, who then performs the guessing attack. In that case, the probability that no teleportation needs corrections is much lower, i.e.\ $1/d^{2(m-1)}$.}:
\begin{itemize}
    \item Verifiers $\mathsf{V_A}, \mathsf{V_B}$ send $d_A, d_B$ dimensional quantum inputs, respectively, to \textsf{P}. They also send classical information $x, y$ (of any size), respectively.
    \item \textsf{P} computes a function $f(x,y)$ and, based on that result, applies a quantum operation $\mathcal{Q}_{f(x,y)}$ to the inputs. This yields two outputs, one intended for $\mathsf{V_A}$ and one for $\mathsf{V_B}$. These outputs are forwarded to the corresponding verifier.
    \item The verifiers check if what they received matches the specific honest protocol and that the responses arrived in time.
\end{itemize}
We will denote this as $\text{QPV}(d_A, d_B, f)$ and the protocol is to be repeated for $n$ rounds, either sequentially or in parallel. Now let $k \coloneqq |\operatorname{Im}(f)|$ be the number of possible quantum operations to be applied at \textsf{P}. It turns out that there always exists a perfect attack consuming at most $O(n\log d)$ (qubit) EPR pairs, where $d = \max\{ d_A, d_B \}$, as long as the loss is high enough. We make this precise in the following statement.
\begin{proposition}
Let $d = \max\{ d_A, d_B \}$ and $k = |\operatorname{Im}(f)|$. Any $n$-round $\emph{QPV}(d_A, d_B, f)$ protocol can be attacked with $\tilde{O}(n)$ EPR pairs if the fraction $\eta$ of rounds that is used for security analysis fulfills $\eta \leq \frac{1}{k d^2}$.
\end{proposition}
\begin{proof}\,
Without loss of generality, let attacker \textsf{A} receive classical input $x$. As soon as \textsf{A} receives his quantum input, quantum teleportation \cite{Teleportation} can be used to teleport the state to \textsf{B} (consuming a $d = \max\{ d_A, d_B \}$ dimensional maximally entangled state\footnote{Attackers do not know the dimension of their local inputs a priori.}), after which \textsf{A} sends to \textsf{B} which teleportation corrections were to apply and the classical information $x$. With probability $p_{00}=1/d^2$ there are no teleportation corrections to apply, in which case \textsf{B} holds both input states locally and before the honest party \textsf{P} would have received them. \textsf{B} can guess the value of $f(x,y)$ and immediately apply the operation \textsf{P} was asked to apply, send the part (e.g.\ a subsystem or a measurement result) that $\mathsf{V_A}$ is supposed to receive to \textsf{A} and keep the part that $\mathsf{V_B}$ is supposed to receive. With probability $1/k$ attacker \textsf{B} guesses the value of $f(x,y)$ correctly. If the quantum state picked up teleportation corrections in the first place or if it turns out that \textsf{B} guessed $f(x,y)$ wrongly (of which both get to know as soon as they receive the communication from the other attacker), attackers deny to answer and both send the corresponding loss symbol `$\varnothing$'. If there were no corrections to apply and \textsf{B} guessed $f(x,y)$ correctly, both send their respective parts that the verifiers are supposed to receive to them. As they are only required to answer $\eta \leq \frac{1}{k d^2}$ of all rounds, they can simply choose those perfect rounds without teleportation corrections and a correct guess for $f(x,y)$. Applying this strategy in each round costs them $n \log d$ (qubit) EPR pairs.
\end{proof}

The statement shines light on another facet of loss tolerance -- the attacker's ability to post-select on a ``correct'' guess after communication. They can always do that, if they pre-share entanglement, by simply guessing the teleportation corrections. In protocols with quantum input from one side and classical information determining the action of \textsf{P}, like \cite{kent_quantum_2011,chakraborty_practical_2015, unruh_quantum_2014,junge2021geometry,bluhm2021position}, attackers can, even without pre-shared entanglement, do this post-selection simply by first guessing $f(x,y)$, then applying the operation $\mathcal{Q}_{f(x,y)}$ on the quantum input and communicating $x,y$ to each other so that both attackers know if the initial guess was correct or not. This is captured in the above bound, identifying all-classical input from one side with $d=1$. In that sense, QPV protocols containing classical input cannot be fully loss-tolerant. It is now also clearer why the protocol in \cite{lim_loss-tolerant_2016} and our QPV$_\textsf{SWAP}$ \cite{allerstorfer2022practical} are fully loss-tolerant if attackers do not pre-share any entanglement. Without shared entanglement, the attackers simply have no way of ever knowing if their guess was correct, because there is no information about it leaving the verifiers.

We conclude that there does not exist a fully loss tolerant QPV protocol of the above type (in particular, if a linear amount of entanglement is pre-shared). The best we can achieve in all generality against entangled attackers is partial loss tolerance. By increasing $k$ and/or $d$ such that $kd^2 > 1/\eta$, however, partial loss tolerance \textit{good enough for all practical purposes} could be achieved. We will therefore need to bound \textit{both} the pre-shared entanglement of the attackers and the loss for a practical QPV protocol. 

\section{Conclusion}

We gave an explicit example of the first QPV protocol in which there is an advantage for attackers who share no entanglement to use quantum communication over classical communication. The protocol depends on determining whether two states were either both symmetric or both antisymmetric. The probability of success under LOCC operations was analytically shown to be equal to 17/18, while attackers with access to quantum communication could attack this protocol perfectly with a single round of simultaneous quantum communication. This suggests that the role of quantum communication can be more important than previously thought. 

Diving into the idea that there can be an advantage to using quantum communication over classical communication we also showed that this separation between quantum communication and classical communication was somewhat constructed. The existence of a quantum communication attack on a protocol that is safe against attackers restricted to classical communication implies the existence of a similar protocol that is safe against quantum communication. In order to prove this we showed that for every quantum communication attack, two new QPV protocols arise. By repeating this argument whenever there was a quantum communication attack, we used the theorem of emergent classicality of channels from \cite{qi2020emergent} to show that the quantum inputs become approximately classical, which in turn implies the existence of a classical communication attack on the original protocol which violated the assumption of the original being secure against classical communication. Thus, ultimately showing that somewhere in the recursion there must have been a protocol that was secure against attackers allowed to use a single round of simultaneous quantum communication.

Setting out to find an explicit example of a protocol that is safe against quantum communication, we prove that the task of Bell state discrimination cannot have a higher success probability than $3/4$ when two attackers are restricted to local operations and a single round of quantum communication. This is higher than the optimal success probability of $1/2$ for attackers restricted to classical communication. We do suspect the quantum communication bound to be lower than $3/4$, but have been unable to make the bound tighter just yet. Nonetheless, this is the first example of a protocol that is fully loss-tolerant against classical communication and stays secure against attackers using quantum communication.

The previous statement immediately brings up the question whether the protocol is also safe against attackers allowed to use quantum communication, and allowed to say loss. We answer this question in the affirmative and prove that there cannot be a strategy that perfectly distinguishes all Bell states. However, we only find a strict inequality, i.e. that $p_{\text{succ}}^{\text{qc}}(\eta) < 1$ for any transmission rate $\eta \in (0,1]$. An interesting follow-up question is whether we can make this bound into a finite gap, just as with the case where the attackers were not allowed to say loss but could use quantum communication. Nonetheless, this is the first example of a QPV protocol that remains fully loss tolerant and secure even in the quantum communication setting. This is of interest since it is the most general setting in the case where we do not allow attackers to pre-share entanglement. 

Extending on the idea of incorporating loss in different settings for attackers we note that in order for QPV to become feasible, not only do we want our protocols to be loss-tolerant against quantum communication, but we also want our protocols to be loss-tolerant against some amount of pre-shared entanglement among the attackers. However, we show that any QPV protocol can be attacked by only a linear amount of entanglement in this setting, given that the loss is high enough. This is based on the simple observation that in such a scenario attackers can post-select on those rounds in which the attempted quantum teleportation did not incur teleportation corrections.

In conclusion, in this paper we investigated the advantage of a single round of simultaneous quantum communication over classical communication in QPV attacks. We find that this advantage can be strictly better than classical communication, but we also find that in some cases it is still strictly worse than the best non-local operations, which proves security for those cases. 

\paragraph{Acknowledgments}
We would like to thank Freek Witteveen, Wolfgang L\"offler and Kirsten Kanneworff for many useful discussions. RA and HB were supported by the Dutch Research Council (NWO/OCW), as part of the Quantum Software Consortium programme (project number 024.003.037). PVL and HB were supported by the Dutch Research Council (NWO/OCW), as part of the NWO Gravitation Programme Networks (project number 024.002.003). Part of this research was finalized at Sorbonne Université. PVL acknowledges support from France 2030 under the French National Research Agency (ANR) award ANR-22-PETQ-0007.

\paragraph{Author contributions}
All authors contributed equally to this work. The authors developed the mathematical proofs themselves. PVL used a large language model to transcribe one of the proofs into \LaTeX{} after he had written it out by hand. AI tools also assisted with proofreading the manuscript.

\pagebreak

\bibliographystyle{alphaurl}
\bibliography{QPV_folder}

\newcommand{\etalchar}[1]{$^{#1}$}
\begin{thebibliography}{BCWDW01}

\bibitem[ABSL22]{allerstorfer2022practical}
Rene Allerstorfer, Harry Buhrman, Florian Speelman, and Philip~Verduyn Lunel.
\newblock Towards practical and error-robust quantum position verification, 2022.
\newblock \href {https://doi.org/10.48550/arXiv.2106.12911} {\path{doi:10.48550/arXiv.2106.12911}}.

\bibitem[BBC{\etalchar{+}}93]{Teleportation}
Charles~H. Bennett, Gilles Brassard, Claude Cr\'epeau, Richard Jozsa, Asher Peres, and William~K. Wootters.
\newblock Teleporting an unknown quantum state via dual classical and {Einstein}-{Podolsky}-{Rosen} channels.
\newblock {\em Phys. Rev. Lett.}, 70:1895--1899, Mar 1993.
\newblock \href {https://doi.org/10.1103/physrevlett.70.1895} {\path{doi:10.1103/physrevlett.70.1895}}.

\bibitem[BCF{\etalchar{+}}11]{buhrman_position-based_2011}
Harry Buhrman, Nishanth Chandran, Serge Fehr, Ran Gelles, Vipul Goyal, Rafail Ostrovsky, and Christian Schaffner.
\newblock Position-{Based} {Quantum} {Cryptography}: {Impossibility} and {Constructions}.
\newblock In {\em Advances in {Cryptology} – {CRYPTO} 2011}, pages 429--446, Berlin, Heidelberg, 2011.
\newblock \href {https://doi.org/10.1007/978-3-642-22792-9_24} {\path{doi:10.1007/978-3-642-22792-9_24}}.

\bibitem[BCS22]{bluhm2021position}
Andreas Bluhm, Matthias Christandl, and Florian Speelman.
\newblock A single-qubit position verification protocol that is secure against multi-qubit attacks.
\newblock {\em Nature Physics}, 18(6):623--626, 2022.
\newblock \href {https://doi.org/10.1038/s41567-022-01577-0} {\path{doi:10.1038/s41567-022-01577-0}}.

\bibitem[BCWDW01]{buhrman_quantum_2001}
Harry Buhrman, Richard Cleve, John Watrous, and Ronald De~Wolf.
\newblock Quantum fingerprinting.
\newblock {\em Physical Review Letters}, 87(16):167902, 2001.
\newblock \href {https://doi.org/10.1103/PhysRevLett.87.167902} {\path{doi:10.1103/PhysRevLett.87.167902}}.

\bibitem[BFSS13]{buhrman2013garden}
Harry Buhrman, Serge Fehr, Christian Schaffner, and Florian Speelman.
\newblock The garden-hose model.
\newblock In {\em Proceedings of the 4th conference on Innovations in Theoretical Computer Science}, pages 145--158, 2013.
\newblock \href {https://doi.org/10.1145/2422436.2422455} {\path{doi:10.1145/2422436.2422455}}.

\bibitem[BK11]{beigi_simplified_2011}
Salman Beigi and Robert Koenig.
\newblock Simplified instantaneous non-local quantum computation with applications to position-based cryptography.
\newblock {\em New Journal of Physics}, 13(9):093036, 2011.
\newblock arXiv: 1101.1065.
\newblock \href {https://doi.org/10.1088/1367-2630/13/9/093036} {\path{doi:10.1088/1367-2630/13/9/093036}}.

\bibitem[BKMS06]{KentPatent2006}
Raymond~G. Beausoleil, Adrian Kent, William~J. Munro, and Timothy~P. Spiller.
\newblock Tagging systems, US patent 7075438, 2006.

\bibitem[BRSdW11]{buhrman_near-optimal_2011}
Harry Buhrman, Oded Regev, Giannicola Scarpa, and Ronald de~Wolf.
\newblock Near-{Optimal} and {Explicit} {Bell} {Inequality} {Violations}.
\newblock In {\em 2011 {IEEE} 26th {Annual} {Conference} on {Computational} {Complexity}}, pages 157--166. IEEE, June 2011.
\newblock \href {https://doi.org/10.1109/CCC.2011.30} {\path{doi:10.1109/CCC.2011.30}}.

\bibitem[CGMO09]{chandran_position_2009}
Nishanth Chandran, Vipul Goyal, Ryan Moriarty, and Rafail Ostrovsky.
\newblock Position {Based} {Cryptography}.
\newblock In {\em Advances in {Cryptology} - {CRYPTO} 2009}, Lecture {Notes} in {Computer} {Science}, pages 391--407. Springer, 2009.
\newblock \href {https://doi.org/10.1007/978-3-642-03356-8_23} {\path{doi:10.1007/978-3-642-03356-8_23}}.

\bibitem[CL15]{chakraborty_practical_2015}
Kaushik Chakraborty and Anthony Leverrier.
\newblock Practical {Position}-{Based} {Quantum} {Cryptography}.
\newblock {\em Physical Review A}, 92(5):052304, November 2015.
\newblock arXiv: 1507.00626.
\newblock \href {https://doi.org/10.1103/PhysRevA.92.052304} {\path{doi:10.1103/PhysRevA.92.052304}}.

\bibitem[CW04]{christandl2004squashed}
Matthias Christandl and Andreas Winter.
\newblock “{Squashed} entanglement”: an additive entanglement measure.
\newblock {\em Journal of mathematical physics}, 45(3):829--840, 2004.
\newblock \href {https://doi.org/10.1063/1.1643788} {\path{doi:10.1063/1.1643788}}.

\bibitem[Dol19]{dolev2019}
Kfir Dolev.
\newblock Constraining the doability of relativistic quantum tasks.
\newblock {\em arXiv preprint arXiv:1909.05403}, 2019.
\newblock \href {https://doi.org/10.48550/arXiv.1909.05403} {\path{doi:10.48550/arXiv.1909.05403}}.

\bibitem[DW05]{devetak2005distillation}
Igor Devetak and Andreas Winter.
\newblock Distillation of secret key and entanglement from quantum states.
\newblock {\em Proceedings of the Royal Society A: Mathematical, Physical and engineering sciences}, 461(2053):207--235, 2005.
\newblock \href {https://doi.org/10.1098/rspa.2004.1372} {\path{doi:10.1098/rspa.2004.1372}}.

\bibitem[GC20]{gonzales_bounds_2020}
Alvin Gonzales and Eric Chitambar.
\newblock Bounds on {Instantaneous} {Nonlocal} {Quantum} {Computation}.
\newblock {\em IEEE Transactions on Information Theory}, 66(5):2951--2963, May 2020.
\newblock \href {https://doi.org/10.1109/TIT.2019.2950190} {\path{doi:10.1109/TIT.2019.2950190}}.

\bibitem[GLW16]{gao2016quantum}
Fei Gao, Bin Liu, and QiaoYan Wen.
\newblock Quantum position verification in bounded-attack-frequency model.
\newblock {\em SCIENCE CHINA Physics, Mechanics \& Astronomy}, 59(11):110311, 2016.
\newblock \href {https://doi.org/10.1007/s11433-016-0234-0} {\path{doi:10.1007/s11433-016-0234-0}}.

\bibitem[HHH99]{HorodeckiImnperfectTeleportation}
Micha\l{} Horodecki, Pawe\l{} Horodecki, and Ryszard Horodecki.
\newblock General teleportation channel, singlet fraction, and quasidistillation.
\newblock {\em Phys. Rev. A}, 60:1888--1898, Sep 1999.
\newblock \href {https://doi.org/10.1103/PhysRevA.60.1888} {\path{doi:10.1103/PhysRevA.60.1888}}.

\bibitem[JKPPG22]{junge2021geometry}
Marius Junge, Aleksander~M Kubicki, Carlos Palazuelos, and David P{\'e}rez-Garc{\'\i}a.
\newblock Geometry of banach spaces: a new route towards position based cryptography.
\newblock {\em Communications in Mathematical Physics}, 394(2):625--678, 2022.
\newblock \href {https://doi.org/10.1007/s00220-022-04407-9} {\path{doi:10.1007/s00220-022-04407-9}}.

\bibitem[KMS11]{kent_quantum_2011}
Adrian Kent, William~J. Munro, and Timothy~P. Spiller.
\newblock Quantum {Tagging}: {Authenticating} {Location} via {Quantum} {Information} and {Relativistic} {Signalling} {Constraints}.
\newblock {\em Physical Review A}, 84(1):012326, July 2011.
\newblock \href {https://doi.org/10.1103/PhysRevA.84.012326} {\path{doi:10.1103/PhysRevA.84.012326}}.

\bibitem[LL11]{lau_insecurity_2011}
Hoi~Kwan Lau and Hoi~Kwong Lo.
\newblock Insecurity of position-based quantum cryptography protocols against entanglement attacks.
\newblock {\em Physical Review A}, 83(1):012322, January 2011.
\newblock \href {https://doi.org/10.1103/PhysRevA.83.012322} {\path{doi:10.1103/PhysRevA.83.012322}}.

\bibitem[LXS{\etalchar{+}}16]{lim_loss-tolerant_2016}
Charles Ci~Wen Lim, Feihu Xu, George Siopsis, Eric Chitambar, Philip~G. Evans, and Bing Qi.
\newblock Loss-tolerant quantum secure positioning with weak laser sources.
\newblock {\em Physical Review A}, 94(3):032315, September 2016.
\newblock \href {https://doi.org/10.1103/PhysRevA.94.032315} {\path{doi:10.1103/PhysRevA.94.032315}}.

\bibitem[Mal10a]{Malaney2010}
Robert~A. Malaney.
\newblock Location-dependent communications using quantum entanglement.
\newblock {\em Phys. Rev. A}, 81:042319, Apr 2010.
\newblock \href {https://doi.org/10.1103/PhysRevA.81.042319} {\path{doi:10.1103/PhysRevA.81.042319}}.

\bibitem[Mal10b]{MalaneyNoisy2010}
Robert~A. Malaney.
\newblock Quantum location verification in noisy channels.
\newblock In {\em 2010 IEEE Global Telecommunications Conference GLOBECOM 2010}, pages 1--6, 2010.
\newblock \href {https://doi.org/10.1109/GLOCOM.2010.5684009} {\path{doi:10.1109/GLOCOM.2010.5684009}}.

\bibitem[OCCG20]{olivo_breaking_2020}
Andrea Olivo, Ulysse Chabaud, André Chailloux, and Frédéric Grosshans.
\newblock Breaking simple quantum position verification protocols with little entanglement.
\newblock {\em arXiv:2007.15808}, July 2020.
\newblock \href {https://doi.org/10.48550/arXiv.2007.15808} {\path{doi:10.48550/arXiv.2007.15808}}.

\bibitem[QLL{\etalchar{+}}15]{qi_free-space_2015}
Bing Qi, Hoi-Kwong Lo, Charles Ci~Wen Lim, George Siopsis, Eric~A. Chitambar, Raphael Pooser, Philip~G. Evans, and Warren Grice.
\newblock Free-space reconfigurable quantum key distribution network.
\newblock {\em 2015 IEEE International Conference on Space Optical Systems and Applications (ICSOS)}, October 2015.
\newblock \href {https://doi.org/10.1109/ICSOS.2015.7425063} {\path{doi:10.1109/ICSOS.2015.7425063}}.

\bibitem[QR21]{qi2020emergent}
Xiao-Liang Qi and Daniel Ranard.
\newblock Emergent classicality in general multipartite states and channels.
\newblock {\em Quantum}, 5:555, 2021.
\newblock \href {https://doi.org/10.22331/q-2021-09-28-555} {\path{doi:10.22331/q-2021-09-28-555}}.

\bibitem[QS15]{qi_loss-tolerant_2015}
Bing Qi and George Siopsis.
\newblock Loss-tolerant position-based quantum cryptography.
\newblock {\em Physical Review A}, 91(4):042337, April 2015.
\newblock \href {https://doi.org/10.1103/PhysRevA.91.042337} {\path{doi:10.1103/PhysRevA.91.042337}}.

\bibitem[RG15]{ribeiro_tight_2015}
Jérémy Ribeiro and Frédéric Grosshans.
\newblock A {Tight} {Lower} {Bound} for the {BB84}-states {Quantum}-{Position}-{Verification} {Protocol}.
\newblock {\em arXiv:1504.07171}, June 2015.
\newblock \href {https://doi.org/10.48550/arXiv.1504.07171} {\path{doi:10.48550/arXiv.1504.07171}}.

\bibitem[Spe16]{speelman2016}
Florian Speelman.
\newblock {Instantaneous Non-Local Computation of Low T-Depth Quantum Circuits}.
\newblock In {\em 11th Conference on the Theory of Quantum Computation, Communication and Cryptography (TQC 2016)}, volume~61, pages 9:1--9:24. Schloss Dagstuhl--Leibniz-Zentrum fuer Informatik, 2016.
\newblock \href {https://doi.org/10.4230/LIPIcs.TQC.2016.9} {\path{doi:10.4230/LIPIcs.TQC.2016.9}}.

\bibitem[TFKW13]{tomamichel_monogamy--entanglement_2013}
Marco Tomamichel, Serge Fehr, Jędrzej Kaniewski, and Stephanie Wehner.
\newblock A {Monogamy}-of-{Entanglement} {Game} {With} {Applications} to {Device}-{Independent} {Quantum} {Cryptography}.
\newblock {\em New Journal of Physics}, 15(10), October 2013.
\newblock \href {https://doi.org/10.1088/1367-2630/15/10/103002} {\path{doi:10.1088/1367-2630/15/10/103002}}.

\bibitem[Unr14]{unruh_quantum_2014}
Dominique Unruh.
\newblock Quantum {Position} {Verification} in the {Random} {Oracle} {Model}.
\newblock In {\em Advances in {Cryptology} – {CRYPTO} 2014}, pages 1--18. Springer Berlin Heidelberg, 2014.
\newblock \href {https://doi.org/10.1007/978-3-662-44381-1_1} {\path{doi:10.1007/978-3-662-44381-1_1}}.

\bibitem[Wat18]{watrous_theory_2018}
John Watrous.
\newblock {\em The {Theory} of {Quantum} {Information}}, page 418.
\newblock Cambridge University Press, 1 edition, April 2018.
\newblock URL: \url{https://www.cambridge.org/core/product/identifier/9781316848142/type/book}, \href {https://doi.org/10.1017/9781316848142} {\path{doi:10.1017/9781316848142}}.

\bibitem[Win99]{winter1999coding}
Andreas Winter.
\newblock Coding theorem and strong converse for quantum channels.
\newblock {\em IEEE Transactions on Information Theory}, 45(7):2481--2485, 1999.
\newblock \href {https://doi.org/10.1109/18.796385} {\path{doi:10.1109/18.796385}}.

\end{thebibliography}

\appendix
\section{Optimal PPT Measurements for \texorpdfstring{QPV$_{\text{Sym/Antisym}}$}{qpvsymmasymm}} \label{AppendixQCbetter}
 
In this section we will solve the SDP program that optimizes the probability of success for two adversaries restricted to LOCC operations of discriminating a random symmetric state from the antisymmetric state. The SDP formulation of this protocol is as follows: 

\begin{minipage}{0.48\textwidth}
\begin{align*}
    &\textbf{Primal Program}\\
    \textbf{maximize: } &\frac{1}{2} \Tr[ \Pi_0 \rho_0 + \Pi_1 \rho_1] \\
    \textbf{subject to: } &\Pi_0 + \Pi_1 = \mathbbm{1}_{2^2} \\
    & \Pi_k \in \text{PPT}(\mathsf{A}:\mathsf{B}), \ \ \ k \in \{0,1\} \\[1ex]
\end{align*}
\end{minipage}
\begin{minipage}{0.48\textwidth}
\begin{align*}
    &\textbf{Dual Program} \\
    \textbf{minimize: } &\Tr[Y] \\
    \textbf{subject to: } &Y - Q^{T_\mathsf{B}}_{i} - \rho_i / 2 \succeq 0, \ \ \ i \in \{0,1\} \\
    & Y\in \text{Herm}(\mathsf{A} \otimes \mathsf{B}) \\
    & Q_i \in \text{Pos}(\mathsf{A}, \mathsf{B}), \ \ \ i \in \{0,1\}. \\
\end{align*}
\end{minipage}
Where we write $\rho_0$ for $\rho_{\text{sym}}$ and $\rho_1$ for $\rho_{\text{antisym}}$ whose respective density matrices are:
\begin{align*}
    &\rho_{0} = \frac{1}{6} \left( \begin{matrix} 2 & 0 & 0 & 0 \\ 0 & 1 & 1 & 0 \\ 0 & 1 & 1 & 0 \\ 0 & 0 & 0 & 2 \end{matrix} \right),  &\rho_{1} = \frac{1}{2} \left( \begin{matrix} 0 & 0 & 0 & 0 \\ 0 & 1 & -1 & 0 \\ 0 & -1 & 1 & 0 \\ 0 & 0 & 0 & 0 \end{matrix} \right).
\end{align*}
A feasible solution for the primal program, is to measure both states in the computational basis and answering the XOR of the measurement outcomes, so $\Pi_0 = \ket{00}\bra{00} + \ket{11}\bra{11}, \Pi_1 = \ket{01}\bra{01} + \ket{10}\bra{10}$. This strategy has success probability $\frac{1}{2}\Tr[\Pi_0 \rho_0 + \Pi_1 \rho_1] = 5/6$. A feasible solution to the the corresponding dual is
\begin{align*}
    Y = \left( \begin{matrix} \frac{1}{6} & 0 & 0 & 0 \\ 0 & \frac{1}{4} & -\frac{1}{12} & 0 \\ 0 & -\frac{1}{12} & \frac{1}{4} & 0 \\ 0 & 0 & 0 & \frac{1}{6} \end{matrix} \right), &&Q_0 = 0 \succeq 0, &&Q_1 = \frac{1}{6}\left( \begin{matrix} 1 & 0 & 0 & 1 \\ 0 & 0 & 0 & 0 \\ 0 & 0 & 0 & 0 \\ 1 & 0 & 0 & 1 \end{matrix} \right) = \frac{1}{3} \ket{\Phi^+}\bra{\Phi^+} \succeq 0.
\end{align*}
Where $\Tr[Y] = 5/6$ is an upper bound of the success probability of the protocol optimized over all PPT measurements. Thus we see that the highest probability of success for adversaries restricted to LOCC operations is $5/6$. \\

Now for the protocol where we double the input rounds, but restrict the inputs to be either both symmetric or antisymmetric states we will show that there is no perfect LOCC attack. The corresponding SDP that optimizes over all PPT strategies looks as follows:

\begin{minipage}{0.48\textwidth}
\begin{align*}
    &\textbf{Primal Program}\\
    \textbf{maximize: } &\frac{1}{2} \Tr[ \Pi_0 (\rho_0 \otimes \rho_0) + \Pi_1 (\rho_1 \otimes \rho_1)] \\
    \textbf{subject to: } &\Pi_0 + \Pi_1 = \mathbbm{1}_{2^4} \\
    & \Pi_k \in \text{PPT}(\mathsf{A}:\mathsf{B}), \ \ \ k \in \{0,1\} \\[1ex]
\end{align*}
\end{minipage}
\begin{minipage}{0.48\textwidth}
\begin{align*}
    &\textbf{Dual Program} \\
    \textbf{minimize: } &\Tr[Y] \\
    \textbf{subject to: } &Y - Q^{T_\mathsf{B}}_{i} - (\rho_i \otimes \rho_i) / 2 \succeq 0, \ \ \ i \in \{0,1\} \\
    & Y\in \text{Herm}(\mathsf{A} \otimes \mathsf{B}) \\
    & Q_i \in \text{Pos}(\mathsf{A}, \mathsf{B}), \ \ \ i \in \{0,1\}. \\
\end{align*}
\end{minipage}
We will show that the following is a feasible solution to the dual
\begin{align*}
    Y = \frac{1}{18} \left( 9 (\rho_0 \otimes \rho_0) + 8 (\rho_1 \otimes \rho_1) \right), && Q_0 = 0 \succeq 0, && Q_1 = \frac{1}{18} \left(  (3\rho_0^{T_\mathsf{B}} \otimes 3\rho_0^{T_\mathsf{B}}) - (\rho_1^{T_\mathsf{B}} \otimes \rho_1^{T_\mathsf{B}}) \right).
\end{align*}
Note that in contrast to the optimizations for QPV$^n_{\text{SWAP}}$ protocols the matrix $Y$ is not equal to  the identity matrix with some factor, but nonetheless it is Hermitian. Secondly the eigenvectors of $3\rho_0^{T_\mathsf{B}}$ and $\rho_1^{T_\mathsf{B}}$ are both the 4 Bell states with eigenvalues $\{3/2,1/2,1/2,1/2\}$ and $\{-1/2,1/2,1/2,1/2\}$ respectively. We see for any of the 16 possible eigenvectors of $Q_1$ that the corresponding eigenvalues will be $0, \frac{1}{9}$ or $\frac{1}{18}$. We conclude that $Q_1 \succeq 0$ since all its eigenvalues are non-negative and it is Hermitian.

For $Q_0$ the first constraint in the dual program becomes:
\begin{align*}
    Y - Q^{T_\mathsf{B}}_{0} - \frac{(\rho_0 \otimes \rho_0)}{2} &= \frac{1}{18} \left( 9 (\rho_0 \otimes \rho_0) + 8 (\rho_1 \otimes \rho_1) \right) - \frac{(\rho_0 \otimes \rho_0)}{2} \\
    &= \frac{4}{9} (\rho_1 \otimes \rho_1) \succeq 0.
\end{align*}
And for $Q_1$ we get:
\begin{align*}
    Y - Q^{T_\mathsf{B}}_{1} - \frac{(\rho_1 \otimes \rho_1)}{2} &= \frac{1}{18} \left( 9 (\rho_0 \otimes \rho_0) + 8 (\rho_1 \otimes \rho_1) \right) - \frac{1}{18} \left( 9 (\rho_0 \otimes \rho_0) - (\rho_1 \otimes \rho_1) \right) - \frac{(\rho_1 \otimes \rho_1)}{2} \\
    &= 0 \succeq 0.
\end{align*}
Thus we have shown that all conditions in the dual program are met, and we get an upper bound on the success probability over all PPT measurements of
\begin{align*}
    \Tr[Y] = \Tr \left[ \frac{1}{18} \left( 9 (\rho_0 \otimes \rho_0) + 8 (\rho_1 \otimes \rho_1) \right) \right] = \frac{17}{18}.
\end{align*}
There is an LOCC strategy that attains this upper bound, namely of applying the single rounds strategy twice, where we only answer antisymmetric if both pairs have unequal measurement. This strategy is correct on all input states except when as an input twice the $\ket{\Psi^+}$ input was sent. This happens with probability 1/18. Thus the best attack for adversaries restricted to LOCC operations has success probability $\frac{17}{18}$, in contrast to adversaries who may use quantum communication for whom there is a perfect attack with success probability $1$ by simply swapping the second register and applying the SWAP test locally.

\end{document}